\documentclass[a4paper, onecolumn, 11pt, accepted=2022-07-04]{quantumarticle}
\pdfoutput=1

\usepackage{tabu}
\usepackage{array}
\usepackage{amsmath}
\usepackage{amsthm}
\usepackage{amssymb}
\usepackage{algorithm}
\usepackage{subfigure}
\usepackage{dirtytalk}
\usepackage{color}
\usepackage[english]{babel}
\usepackage{graphicx}
\usepackage{grffile}
\usepackage{wrapfig,epsfig}
\usepackage{epstopdf}
\usepackage{url}
\usepackage{color}
\usepackage{epstopdf}
\usepackage{algpseudocode}
\usepackage[T1]{fontenc}
\usepackage{bbm}
\usepackage{comment}
\usepackage{dsfont}
\usepackage{mathtools}
\usepackage{enumitem}
\usepackage{tikz}
\usetikzlibrary{backgrounds,fit,decorations.pathreplacing,calc}
\usepackage{qcircuit}
\let\C\relax
\usepackage{hyperref}
\usetikzlibrary{arrows}
\usepackage[margin=1in]{geometry}
\graphicspath{{./fig/}}
\usepackage[nameinlink,capitalize]{cleveref}
\usepackage{textcomp}
\usepackage{multirow}

\newtheorem{theorem}{Theorem}[section]
\newtheorem{lemma}[theorem]{Lemma}
\newtheorem{definition}[theorem]{Definition}

\newtheorem{corollary}[theorem]{Corollary}

\newtheorem{remark}[theorem]{Remark}
\newtheorem{claim}[theorem]{Claim}

\newtheorem{problem}[theorem]{Problem}

\newcommand{\wh}{\widehat}
\newcommand{\wt}{\widetilde}
\newcommand{\ov}{\overline}

\newcommand{\R}{\mathbb{R}}

\renewcommand{\d}{\mathrm{d}}
\renewcommand{\tilde}{\wt}
\renewcommand{\hat}{\wh}

\renewcommand{\d}{\mathrm{d}}

\DeclareMathOperator*{\E}{{\mathbb{E}}}
\DeclareMathOperator*{\Var}{{\bf {Var}}}

\DeclareMathOperator*{\Z}{\mathbb{Z}}
\DeclareMathOperator*{\C}{\mathbb{C}}

\newcommand{\ket}[1]{\left| #1 \right\rangle}
\newcommand{\bra}[1]{\left\langle #1 \right|}

\DeclareMathOperator{\poly}{poly}

\definecolor{mygreen}{RGB}{80,180,0}
\definecolor{b2}{RGB}{51,153,255}

\newcommand{\nc}{\newcommand}

\nc{\nnl}{\nn \\ &}  
\nc{\fot}{\frac{1}{2}} 
\nc{\oo}[1]{\frac{1}{#1}} 
\newcommand{\ben}{\begin{enumerate}}
\newcommand{\een}{\end{enumerate}}
\nc{\mc}{\mathcal}

\nc{\onenorm}[1]{\L\| #1 \R\|_1} 

\nc{\Ra}{\Rightarrow}
\nc{\zo}{\{0,1\}}

\mathchardef\mhyphen="2D 

\title{Computing Ground State Properties with Early Fault-Tolerant Quantum Computers}

\author{Ruizhe Zhang}
\affiliation{Department of Computer Science, The University of Texas at Austin, Austin, TX 78712, USA.}
\email{ruizhe@utexas.edu}
\author{Guoming Wang}
\affiliation{Zapata Computing Inc., Boston, MA 02110, USA.}
\email{guoming.wang@zapatacomputing.com}
\author{Peter Johnson}
\affiliation{Zapata Computing Inc., Boston, MA 02110, USA.}
\email{peter@zapatacomputing.com}

\begin{document}

\maketitle
\begin{abstract}
Significant effort in applied quantum computing has been devoted to the problem of ground state energy estimation for molecules and materials.
Yet, for many applications of practical value, additional properties of the ground state must be estimated.
These include Green's functions used to compute electron transport in materials and the one-particle reduced density matrices used to compute electric dipoles of molecules.
In this paper, we propose a quantum-classical hybrid algorithm to efficiently estimate such ground state properties with high accuracy using low-depth quantum circuits.
We provide an analysis of various costs (circuit repetitions, maximal evolution time, and expected total runtime) as a function of target accuracy, spectral gap, and initial ground state overlap.
This algorithm suggests a concrete approach to using early fault tolerant quantum computers for carrying out industry-relevant molecular and materials calculations.
\end{abstract}

\newpage
\section{Introduction}

One of the primary applications of quantum computing is the simulation of materials and molecules, which are inherently quantum mechanical.
It is hoped that future powerful quantum computers will be used in the development of materials and drug discovery \cite{cao2018potential}.
Although they have yet to realize commercial application, quantum computers have been improving at a rapid rate, increasing the demand for quantum algorithms with high-impact use cases.
To date, the main focus of quantum algorithm development for quantum chemistry and materials has been on ground state energy estimation \cite{cao2019quantum}.
This problem is mathematically formulated as estimating the lowest eigenvalue of the Hamiltonian matrix that characterizes the physical system.
One of the first quantum chemistry applications of quantum computers was to use quantum phase estimation for estimating the ground state energy of small molecules \cite{aspuru2005simulated}.
More recently, the variational quantum eigensolver algorithm \cite{peruzzo2014variational} was developed to use near-term intermediate-scale quantum (NISQ) computers to solve the ground state energy estimation problem.


However, in characterizing materials or analyzing small molecules for drug discovery, one often needs to estimate properties of the ground state beyond just the energy. 
These include transport properties \cite{meir1992landauer}, electric dipole moments \cite{jensen2017introduction}, and molecular forces \cite{o2019calculating}.
Such properties depend on
expectation values of observables $O$ with respect to the ground state of a Hamiltonian $H$.
The problem of estimating such quantities was studied in \cite{amb14, gy19, gpy20}, showing that it is even harder, in a complexity theoretic sense, than the ground state energy estimation problem in general. 
A straightforward approach to estimating ground state properties is to first (approximately) prepare the ground state, from which properties can be estimated. Many algorithms (e.g. \cite{pw09,gtc19,lt20}) have been developed for ground state preparation.  
However, these algorithms only work for idealistic quantum computers, and the quantum circuit depths involved in these methods are too deep to even be implemented on early fault-tolerant quantum computers.
Another approach to preparing ground states that is more amenable to near-term quantum computers is to use the variational quantum eigensolver algorithm \cite{mcardle2019digital, o2019calculating}.
However, recent work has suggested that VQE alone is not practical for solving problems of industrial relevance \cite{gonthier2020identifying}; 
estimation methods which are more efficient (e.g. \cite{wang2021minimizing}) than prepare and measure estimation, as used in VQE,
seem necessary in order for quantum computers to compete with state-of-the-art methods in quantum chemistry and materials.
Further issues with the variational quantum eigensolver and its variants are that there are no guarantees on the quality of the output ground state and that heuristic optimization methods struggle to prepare high-fidelity ground states.

This motivates the development of quantum algorithms for ground state property estimation (GSPE) which are both reliable \emph{and} able to be run on near-term quantum computers (e.g. early fault-tolerant quantum devices) with the following characteristics: (1) The circuit depth (or the maximal Hamiltonian evolution time) is small even with the price of increasing the total circuit size (or evolution time). (2) The number of logical qubits is limited. The early fault-tolerant model captures the challenges of building a large-scale long-time coherent quantum device, while also being able to solve many important problems with provable performance guarantees \cite{bmn21,bom21,cam21,lt21,lay22}.
The central question that this paper addresses is then:
\begin{center}
\emph{Is it possible to estimate ground state properties of a Hamiltonian reliably using early fault-tolerant quantum computers?}
\end{center}




In this paper, we provide an affirmative answer to this question. Furthermore, we propose an algorithm for the ground state property estimation using low-depth quantum circuits. The main theorem is stated as follows:

\begin{theorem}[Main theorem, informal]\label{thm:intro_app_sim}
Given a Hamiltonian $H$ and an observable $O$. Suppose we have access to a unitary $U_I$ that prepares a state $\ket{\phi_0}$ that has non-trivial overlap with the ground state $\ket{\psi_0}$ of $H$. Then, there exists an algorithm to estimate $\bra{\psi_0}O\ket{\psi_0}$ with high accuracy and low-depth: the maximal Hamiltonian evolution time is $\wt{O}(\gamma^{-1})$, where $\gamma$ is the spectral gap of $H$.
\end{theorem}

We make a few remarks about our main result. First, we note that the maximal evolution time, which is the maximal length of time we need to perform coherent time evolution, can roughly determine the depth of the quantum circuit. Our result achieves a nearly-linear dependence on $\gamma^{-1}$ and only poly-logarithmic on the accuracy $\epsilon^{-1}$, which improves the $\wt{O}(\epsilon^{-1})$ maximal evolution time in the ground state energy estimation algorithms \cite{som19,lt21,cbk21,ral21}. Second, our result does not violate the Heisenberg limit because the total evolution time still depends on $\poly(\epsilon^{-1})$.  Third, similar to almost all prior works in ground state preparation and energy estimation (e.g. \cite{som19,lt20,lt21}), we need the assumption that the initial state has some \emph{nontrivial overlap} with the ground state, as otherwise the problem will become computationally intractable. Last, we consider the Hamiltonian as a black-box, which is a common model in this field. To implement our algorithm, for sparse local Hamiltonian, we can use the current state-of-the-art Hamiltonian simulation methods \cite{bcc15,lc17,cmn18,lc19} with gate complexity depending linearly in the evolution time and logarithmically in the accuracy. 

\paragraph{Comparison to the straightforward method. }We can compare our algorithm with the straightforward approach of GSPE that first prepares the ground state and then applies quantum phase estimation (QPE) to estimate the ground state property. 
\begin{itemize}
\item In the first step, to achieve an $\epsilon$-accuracy for the estimation, the ground state need to be prepared with fidelity at least $1-\epsilon$ using the methods in \cite{gtc19,lt20}, which have circuit depth $\wt{O}(\gamma^{-1}\eta^{-1})$ where $\eta$ is the overlap between the initial state and the ground state. 

\item In the second step, QPE \cite{kos07,ral21} requires circuit depth $\wt{O}(\epsilon^{-1})$ for an $\epsilon$-accuracy estimation for the ground state property. 
\end{itemize}
Therefore, this straightforward approach has circuit depth $\wt{O}(\gamma^{-1}\eta^{-1}+\epsilon^{-1})$, while our algorithm has circuit depth $\wt{O}(\gamma^{-1})$. Furthermore, they also need many (i.e., $\omega(1)$) additional ancilla qubits for preparing the ground state, while we only use one ancilla qubit. 
Our algorithm has a great advantage when the Hamiltonian's spectral gap is much larger than the estimation accuracy, making it easier to be implemented in the early fault-tolerant devices.

\paragraph{Organization.}
In Section \ref{sec:gspe} we formally state the problem of ground state property estimation.
In Section \ref{sec:gsee} we review the method developed in \cite{lt21} for estimating ground state energies.
In the next three sections we explain our main algorithms and give an analysis for their performances starting from the simplest case and building to the most-involved, general case.
Section \ref{sec:commutative_alg} presents the case of a unitary observable which commutes with the Hamiltonian.
Section \ref{sec:unitary_alg} presents the case of a unitary observable which does not necessarily commute with the Hamiltonian.
Section \ref{sec:general_alg} describes the case of a general observable.
Then, Section \ref{sec:apps} gives two applications of the ground state property estimation algorithm. Section~\ref{sec:discuss} gives a discussion of the results and presents some open questions. 
\section{Ground State Property Estimation Problem}
\label{sec:gspe}
In this section, we will formally define the ground state property estimation problem. This problem was initially studied by Ambainis \cite{amb14} as the approximate simulation problem (\textsf{APX-SIM}), and he proved that \textsf{APX-SIM} is $\mathsf{P^{QMA[\log]}}$-complete\footnote{$\mathsf{P^{QMA[\log]}}$ contains the problems with polynomial-time classical algorithms that are allowed to make $O(\log n)$ queries to an oracle solving a promise problem in \textsf{QMA}.}. 
\begin{problem}[Approximate simulation (\textsf{APX-SIM}), \cite{amb14}]\label{prob:app_sim}
Given a $k$-local Hamiltonian $H$, an $\ell$-local observable $O$, and real numbers $a,b,\epsilon$ such that $b-a\geq 1/\poly(n)$, and $\epsilon\geq 1/\poly(n)$, for $n$ the number of qubits the Hamiltonian $H$ acts on, decide:
\begin{itemize}
    \item \textbf{Yes case:} $H$ has a ground state $\ket{\psi_0}$ such that $\bra{\psi_0}O\ket{\psi_0}\leq a$,
    \item \textbf{No case:} for any state $\ket{\psi}$ with $\bra{\psi}H\ket{\psi}\leq \lambda_0 + \epsilon$ where $\lambda_0$ is the ground state energy of $H$, it holds that $\bra{\psi_0}O\ket{\psi_0}\geq b$.
\end{itemize}
\end{problem}

In the follow-up works, \textsf{APX-SIM} was shown to be $\mathsf{P^{QMA[\log]}}$-complete even for 5-local Hamiltonian and 1-local observable \cite{gy19}, and also for some physics models like 2D Heisenberg model and 1D nearest-neighbor, translationally invariant model \cite{gpy20,wbg20}. However, these previous studies only focused on the decision version of this problem. For the purpose of designing efficient algorithms, we first define the ``search version'' of \textsf{APX-SIM} as follows: 
\begin{problem}[Search version of \textsf{APX-SIM}]\label{prob:apx_sim_search}
Given a Hamiltonian $H$, an (local) observable $O$, and $\epsilon\in (0, 1)$, with $\Omega(1)$ probability, estimate $\bra{\psi_0} O \ket{\psi_0}$ with an additive/multiplicative error at most $\epsilon$.
\end{problem}

In general, Problem~\ref{prob:apx_sim_search} will not be more tractable than Problem~\ref{prob:app_sim}. Thus, we may need some prior information about the Hamiltonian $H$ and its ground state. Motivated by the widely used variational quantum eigensolver (VQE) \cite{pmsy14,mrba16} and the Hartree-Fock method \cite{so12} in quantum chemistry, it is often the case that for many real-world Hamiltonians, we are able to efficiently prepare an initial state $\ket{\phi_0}$ that has a nontrivial overlap with the ground state. Moreover, we assume that the Hamiltonian $H$ has a nontrivial spectral gap, where a large family of Hamiltonians in practice satisfy this condition. With these assumptions, we formally define the ground state property estimation problem as follows:

\begin{problem}[Ground state property estimation (GSPE)]\label{prob:gs_prop_est}
Given a Hamiltonian $H$ with spectral gap $\gamma$ and ground state $\ket{\psi_0}$, an observable $O$, a unitary $U_I$ such that it prepares an initial state $\ket{\phi_0}$ with $|\langle \phi_0|\psi_0\rangle|^2\geq \eta$, and $\epsilon\in (0,1)$, estimate $\bra{\psi_0} O\ket{\psi_0}$ with an additive/multiplicative error at most $\epsilon$ with $\Omega(1)$ probability.
\end{problem}

\begin{remark}
We notice that when $O=H$, Problem~\ref{prob:gs_prop_est} becomes the ground state energy estimation problem. Moreover, the prior knowledge of a large overlap for the initial state is required for all quantum algorithms with provable performance guarantees (e.g. \cite{gtc19,lt20,lt21}). It is also worth noting that even with these assumptions, it is unlikely to use a purely classical algorithm to estimate the ground state energy or property to high precision (unless $\mathsf{P}=\mathsf{BQP}$)~\cite{gg21}.
\end{remark}

We propose a high-accuracy, early fault-tolerant quantum algorithm for GSPE that satisfies the following properties:
\begin{itemize}
    \item The maximal evolution time depends \emph{logarithmically} on the accuracy $\epsilon$ and overlap $\eta$.
    \item In addition to the Hamiltonian evolution and observable implementation, it only uses one additional ancilla qubit.
\end{itemize}
\section{An Overview of the Low-Depth Ground State Energy Estimation}
\label{sec:gsee}
In this section, we provide a brief overview of the low-depth ground state energy estimation algorithm proposed by Lin and Tong \cite{lt21}. Our algorithms are inspired by this algorithm and use it as a subroutine.

More specifically, they showed that:
\begin{theorem}[\cite{lt21}]\label{thm:lt21_main}
Given a Hamiltonian $H$ with eigenvalues in the interval $[-\pi/3, \pi /3]$ and its ground state $\ket{\psi_0}$ has energy $\lambda_0$. And suppose we can prepare an initial state $\ket{\phi_0}$ such that $p_0\geq \eta$ for some known $\eta$, where $p_0:=|\langle \phi_0 | \psi_0\rangle|^2$. Then, for any $\epsilon, \nu\in (0, 1)$, there exists an algorithm that estimates $\lambda_0$ with an additive error $\epsilon$ with probability $1-\nu$, by running a parameterized quantum circuit with the maximum quantum evolution time $\wt{O}(\epsilon^{-1})$ and the expected total quantum evolution time $\wt{O}(\epsilon^{-1}\eta^{-2})$.
\end{theorem}
The pseudo-code of their algorithm is given in Algorithm~\ref{alg:gs_energy}.

\begin{algorithm}[ht]
\caption{Ground State Energy Estimation}
\label{alg:gs_energy} 
\begin{algorithmic}[1]
    \algrenewcommand\algorithmicprocedure{\textbf{procedure}}
	\Procedure{EstimateGSE}{$\epsilon,\tau, \eta, \nu$}
	    \State \Comment{Initialization}
	    \State $d\gets O(\delta^{-1}\log(\delta^{-1}\eta^{-1}))$, $\delta \gets \tau \epsilon$\label{ln:set_d}
	    \For{$i\gets -d,\dots,d$}
	        \State $\hat{F}_i\gets \hat{F}_{d,\delta,i}$
	        \State Compute $\theta_i$, the phase angle of $\hat{F}_i$
	    \EndFor
	    \State ${\cal F}\gets \sum_{|i|\leq d}|\hat{F}_i|$
	    \State $N_b\gets \Omega(\log(1/\nu)+\log\log(1/\delta))$, $N_s\gets O(\eta^{-2}\log^2(d))$
	    \State \Comment{Sampling from the quantum circuit}
	    \For{$k\gets 1,\dots,N_bN_s$}
	        \State Independently sample $J_k\sim [-d,d]$ with $\Pr[J_k=j]\propto |\hat{F}_j|$
	        \State Measure $(X_k,Y_k)$ by running the quantum circuit with (Figure.~\ref{fig:hadamard_test}) parameter $k$ 
	        \State $Z_k\gets X_k + iY_k$ 
	    \EndFor
	    \State \Comment{Classical post-processing}
        \State $x_{L}\gets -\pi/3$, $X_{R}\gets \pi/3$
        \While{$x_R - x_L > 2\delta$}\label{ln:inv_cdf_while}\Comment{Invert CDF}
            \State $x_M\gets (x_L + x_R)/2$
            \For{$r\gets 1,\dots, N_b$}
                \State $\ov{G}_r\gets \frac{{\cal F}}{N_s}\sum_{k=(r-1)N_s + 1}^{rN_s} Z_k e^{i(\theta_{J_k} + J_k x_M)}$\Comment{Multi-level Monte Carlo method}\label{ln:mlmc}
            \EndFor
            \If{$|\{r:\ov{G}_r\geq (3/4)\eta\}|\leq N_b / 2$}
                \State $x_R\gets x_M + (2/3)\delta$
            \Else 
                \State $x_L\gets x_M - (2/3)\delta$
            \EndIf
        \EndWhile
        \State \Return $(x_L + x_R)/2$
	\EndProcedure
\end{algorithmic}
\end{algorithm}

The main technique of their algorithm is a classical post-processing procedure that extracts information from the following Hadamard test circuit (Figure~\ref{fig:hadamard_test}).
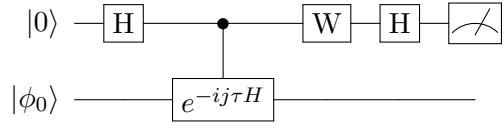
\begin{figure}[H]
    \centering
		\begin{displaymath}
		\Qcircuit @C=1.0em @R=1.2em {
			& & & &\\			
			\lstick{\ket{0}}
			&\gate{\mathrm{H}}	 &\ctrl{1}	& \gate{\mathrm{W}}
			& \gate{\mathrm{H}}			&\meter\\
			\lstick{\ket{\phi_0}} 	 & \qw & \gate{e^{-ij\tau H}} 		 
			&\qw &\qw &\qw
		}		
		\end{displaymath}		
\caption{Quantum circuit parameterized by $j$. $\mathrm{H}$ is the Hadamard gate and $\mathrm{W}$ is either $I$ or a phase gate. A detailed analysis of this circuit is given in Appendix~\ref{sec:quantum_hadamard}.}
    \label{fig:hadamard_test}
\end{figure}

Let the initial state $\ket{\phi_0}$ be expanded as $\ket{\phi_0}=\sum_k \alpha_k \ket{\psi_k}$ in the eigen-basis of $H$ and let $p_k:=|\alpha_k|^2$ be the overlap with the $k$-th eigenstate. They considered the overlaps $p_0,p_1,\dots$ as a density function:
\begin{align}
    p(x):=\sum_k p_k \delta(x- \lambda_k)~~~\forall x\in [-\pi,\pi].
\end{align}
Then, the cumulative distribution function (CDF) $C(x):=\int_{-\pi}^x p(t)\d t$ can be expressed as a convolution of $p(x)$ and the $2\pi$-periodic Heaviside function $H(x)$, which is 0 in $[(2k-1)\pi, 2k\pi)$ and  1 in $[2k\pi, (2k+1)\pi)$ for any $k\in \Z$. Thus, $C(x)$ is also a periodic function, which makes it convenient to apply the Fourier approximation. They showed that $H(x)$ can be approximated by a low-Fourier degree function $F(x)$ in the intervals $[-\pi+\delta, -\delta]$ and $[\delta, \pi-\delta]$. Then, they defined the approximated cumulative distribution function (ACDF) as $\wt{C}(x) := (F \star p)(x)$ and proved that 
\begin{align}
    C(x-\delta) -\eta/8 \leq \wt{C}(x) \leq C(x+\delta) + \eta/8~~~\forall x\in [-\pi/3, \pi/3].
\end{align}
Moreover, for each $x$, we have
\begin{align}
    \wt{C}(x) = \sum_{|j|\leq d} \hat{F}_j e^{ijx} \cdot \bra{\phi_0}e^{-ijH}\ket{\phi_0},
\end{align}
where $\hat{F}_j$ is the Fourier coefficient of $F(x)$. Note that $\bra{\phi_0} e^{-ijH}\ket{\phi_0}$ can be estimated via the parameterized quantum circuit (Figure~\ref{fig:hadamard_test}). Hence, we can estimate the ACDF at every point in $[-\pi/3, \pi/3]$. Moreover, they showed that the multi-level Monte Carlo method can be applied here to save the number of samples needed to achieve a high-accuracy estimation (Line~\ref{ln:mlmc}).

Therefore, we can estimate the ground state energy $\lambda_0$ by locating the first non-zero point of the CDF $C(x)$, which is $\eta/8$-approximated by the ACDF $\wt{C}(x)$. Since we assume that $p_0\geq \eta$, the approximation error and the estimation error of $\wt{C}(x)$ can be tolerated, and we can find $\lambda_0$ via a robust binary search (Line~\ref{ln:inv_cdf_while}). 

We note that the maximal evolution time of this algorithm corresponds to the Fourier degree of $F(x)$, which is $\wt{O}(\epsilon^{-1})$ by the construction, making their algorithm suitable for early fault-tolerant quantum devices.
More details of this algorithm and the proofs are given in Appendix~\ref{sec:lt21_details}.
\section{Algorithm for Commutative Case}
\label{sec:commutative_alg}
In this section, we consider a easier case that $O$ is unitary and commutes with the Hamiltonian $H$, and give a two-step quantum-classical hybrid algorithm for Problem~\ref{prob:gs_prop_est}. More specifically, suppose the initial state can be expanded in the eigenbasis as follows: $\ket{\phi_0}=\sum_{k} c_k \ket{\psi_k}$ with $p_k:=|c_k|^2$. We note that $\{\ket{\psi_k}\}$ is also an eigenbasis of $O$ since $O$ and $H$ commute. In Step 1, we run \cite{lt21}'s algorithm to estimate the ground state energy $\lambda_0$ and the overlap between the initial state and the ground state $p_0$. In Step 2, we construct a similar CDF function for the density $\sum_k O_k p_k \delta(x-\lambda_k)$, where $O_k := \bra{\psi_k} O\ket{\psi_k}$. If we evaluate the CDF at $\lambda_0$, we can obtain an estimate of $O_0$.

\subsection{Step 1: estimate the initial overlap}\label{sec:est_overlap} 
We first run the procedure $\textsc{EstimateGSE}$ (Algorithm~\ref{alg:gs_energy}) to estimate the ground state energy $\lambda_0$ with an additive error $\epsilon$. Let $x^\star$ be the output.
We remark that $x^\star$ satisfy $C(x^\star + \tau \epsilon)\geq p_0$ and $C(x^\star - \tau \epsilon) = 0$. However, we can only extract $p_0$ from the ACDF $\wt{C}(x)$, which satisfies:
\begin{align}
    C(x-\tau\epsilon) -\eta/8 \leq \wt{C}(x)\leq C(x+\tau\epsilon) + \eta/8~~~\forall x\in [-\pi/3, \pi/3].
\end{align}
If $[x-\tau \epsilon, x+\tau\epsilon]$ contains a ``jump'' of $C(x)$, i.e., an eigenvalue $\lambda_k$, then the approximation error of $\wt{C}(x)$ will be large.

Hence, we say a point $x$ is ``good'' for $\lambda_k$ if $[x-\tau \epsilon, x+\tau\epsilon]$ is contained in $[\tau \lambda_k, \tau \lambda_{k+1})$. It is easy to see that $\wt{C}(x)$ will be an $\eta/8$-additive approximation of $\sum_{j\leq k}p_k$ if $x$ is good. Our goal is to find an $x_{\mathsf{good}}$ that is good for $\lambda_0$, and estimating $\wt{C}(x_{\mathsf{good}})$ gives the overlap $p_0$. The following claim gives a way to construct $x_{\mathsf{good}}$ using the spectral gap of $H$.
\begin{claim}[Construct $x_{\mathsf{good}}$]\label{clm:good_x}
Let $\gamma$ be the spectral gap of the Hamiltonian $H$. For any  $\epsilon\in (0, \gamma /4)$, $x^\star + \tau \gamma /2$ is good for $\lambda_0$, where $x^\star$ is the output of $\textsc{EstimateGSE}(\epsilon, \eta)$ (Algorithm~\ref{alg:gs_energy}).
\end{claim}
\begin{proof}
We know that $x^\star$ satisfies:
\begin{align}
    x^\star - \tau \epsilon < \tau\lambda_0 \leq x^\star + \tau \epsilon.
\end{align}
Then, we have
\begin{align}
    x^\star + \tau \gamma/2 > \tau\lambda_0 - \tau \epsilon + \tau \gamma/2 > \tau \lambda_0+\tau \epsilon.
\end{align}
We also have
\begin{align}
    x^\star + \tau \gamma/2 < &~ \tau \lambda_0 + \tau \epsilon + \tau \gamma/2 \\
    \leq &~ \tau (\lambda_1 - \gamma) + \tau \epsilon + \tau \gamma /2\notag\\
    = &~ \tau \lambda_1 + \tau (\epsilon - \gamma/2)\notag\\
    < &~ \tau \lambda_1-\tau \epsilon.
\end{align}
Therefore, $x^\star$ is good for $\lambda_0$.
\end{proof}

We note that in \cite{lt21}, the ACDF's approximation error is chosen to be $\eta/8$. We may directly change it to $\epsilon\eta/8$ without significantly changing the circuit depth, since by Lemma~\ref{lem:approx_Heaviside} the degree of $F$ can only blowup by a log factor of $\epsilon$.

\begin{lemma}[Estimating the overlap]\label{lem:est_overlap}
For any $\epsilon_0,\nu\in (0, 1)$, the overlap $p_0:=|\langle \phi_0|\psi_0\rangle|^2$ can be estimated with multiplicative error $1\pm O(\epsilon_0)$ with probability $1-\nu$ by running the quantum circuit (Figure~\ref{fig:hadamard_test}) $\wt{O}(\epsilon_0^{-2}\eta^{-2})$ times with expected total evolution time $\wt{O}(\gamma^{-1}\epsilon^{-2}\eta^{-2})$ and maximal evolution time $O(\gamma^{-1})$.
\end{lemma}
\begin{proof}
By Claim~\ref{clm:good_x}, if we set the additive error of ground state energy $\lambda_0$ to be $O(\gamma)$, then we can construct an  $x_{\mathsf{good}}$ that is good for $\lambda_0$. By Theorem~\ref{thm:lt21_main}, it can be done with maximum quantum evolution time $\wt{O}(\gamma^{-1})$ and the expected total quantum evolution time $\wt{O}(\gamma^{-1}\eta^{-2})$. Notice that we need to take $d=O(\delta^{-1}\log(\delta^{-1}\epsilon_0^{-1}\eta^{-1}))$ (Line~\ref{ln:set_d} in Algorithm~\ref{alg:gs_energy}) to make $\wt{C}(x_{\mathsf{good}})$ be an $O(\epsilon_0 \eta)$-approximation of $p_0$, where $\delta = \tau \gamma$. 

Next, we estimate $\wt{C}(x_{\mathsf{good}})$ with additive error $\eta \epsilon$ with probability $1-\nu$. We have an unbiased estimator 
\begin{align}
    \ov{G}(x; \mathbf{Z,J})={\cal F}\mathbf{Z}e^{i\theta_{\mathbf{J}}+\mathbf{J}x} 
\end{align}
for $\wt{C}(x)$, where ${\bf Z}:=X+iY$ is measured from the Hadamard test, and ${\bf J}$ is a random variable for the Hamiltonian evolution time sampled proportional to the Fourier weight of $F$, i.e., $\Pr[{\bf J}=j]=|\hat{F}_j|/{\cal F}$ for $-d\leq j\leq d$ and ${\cal F}:=\sum_{|j|\leq d}|\hat{F}_j|$.   

We can show that $\ov{G}(x; \mathbf{Z,J})$ has variance $O(\log^2(d))$, and one estimate can be obtained with evolution time $\wt{O}(\tau d/\log (d))$ in expectation. If we repeatedly sample $\ov{G}(x; \mathbf{Z,J})$ and take the mean of them, then by Chebyshev's inequality, the sample complexity is $\wt{O}(\epsilon_0^{-2}\eta^{-2} \nu^{-2})$ to have an additive error $O(\epsilon_0\eta)$ with probability $1-\nu$. 

Instead, we can use the so-called ``median-of-means'' trick to reduce the sample complexity. More specifically, let $N_g=O(\log(1/\nu))$ and $K=O(\epsilon_0^{-2})$. We first partition $m=N_gK$ samples $(Z_1, J_1),\dots,(Z_{m},J_m)$ into $N_g$ groups of size $K$. Then, for any $i\in [N_g]$, the $i$-th group mean is 
\begin{align}
    \ov{G}_i := \frac{1}{K}\sum_{j=1}^{K} \ov{G}(x; Z_{(i-1)K + j}, J_{(i-1)K + j}).
\end{align}
The final estimator is given by the median of these group means, i.e.,
\begin{align}
    \ov{G}(x):=\mathrm{median}(\ov{G}_1,\dots,\ov{G}_{N_g}).
\end{align}
By Chernoff bound, it is easy to see that $\ov{G}(x)$ has an additive error at most $(\eta\epsilon_0)$ with probability $1-\nu$. It will imply that multiplicative error is at most $1\pm O(\epsilon_0)$ since $p_0=\Theta(\eta)$. And the sample complexity of $\ov{G}(x)$ is $\wt{O}(\epsilon_0^{-2}\eta^{-2})$.
Hence, the expected total evolution time is $\wt{O}(\gamma^{-1}\epsilon_0^{-2}\eta^{-2})$. Since we run the same quantum circuit to estimate $\ov{G}(x)$, the maximal evolution time is still $\wt{O}(\gamma^{-1})$.
\end{proof}

\subsection{Step 2: estimate the $O$-weighted CDF}

To estimate the expectation value of $O$, consider the following quantum circuit:
\begin{figure}[ht]
    \centering
		\begin{displaymath}
		\Qcircuit @C=1.0em @R=1.2em {
			& & & &\\			
			\lstick{\ket{0}}
			&\gate{\mathrm{H}}	 &\ctrl{1}	& \ctrl{1} & \gate{\mathrm{W}}
			& \gate{\mathrm{H}}			&\meter\\
			\lstick{\ket{\phi_0}} 	 & \qw & \gate{e^{-ij\tau H}} & \gate{O}		 
			&\qw &\qw &\qw
		}		
		\end{displaymath}		
    \caption{Quantum circuit parameterized by $j$. $\mathrm{H}$ is the Hadamard gate and $\mathrm{W}$ is either $I$ or a phase gate $S$.}
    \label{fig:hadamard_test_o}
\end{figure}
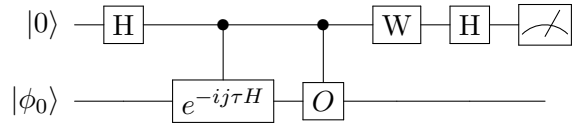

Define the random variables $X_j, Y_j$ be as follows: for $W=I$, $X_j:=1$ if the outcome is 0, and $X_j:=-1$ if the outcome is 1. For $W=S$, $Y_j:=-1$ if the outcome is 0, and $Y_j:=1$ if the outcome is 1. 

Then, we have the following claim on the expectation of the random variables $X_j,Y_j$:
\begin{claim}[A variant of Hadamard test]\label{clm:estimator_expectation_observable}
For any $j\in \Z$, the random variable $X_j + i Y_j$ is an un-biased estimator for $\bra{\phi_0} O e^{-ij\tau H}\ket{\phi_0}$.
\end{claim}
The proof is deferred to Appendix~\ref{sec:quantum_hadamard}.

We can expand $\bra{\phi_0}O e^{-ij\tau H} \ket{\phi_0}$ in the eigenbasis of $H$ (which is also an eigenbasis of $O$):
\begin{align}
    \bra{\phi_0}O e^{-ij\tau H} \ket{\phi_0} = &~ \sum_{k,k'} c_k^* c_{k'}e^{-ij\tau \lambda_k} \bra{\psi_k} O \ket{\psi_k'}\notag\\
    =&~ \sum_k p_k O_k e^{-ij\tau \lambda_k}, 
\end{align}
where the last step follows from the simultaneous diagonalization of $O$ and $H$, and $O_k:=\bra{\psi_k} O \ket{\psi_k}$. We may assume that $|O_k|\leq 1$ for any $k\in \mathbb{N}$.

Inspired by the ground state energy estimation algorithm in \cite{lt21}, we define the $O$-weighted ``density function'' for the observable as follows:
\begin{align}
    p_O(x) := \sum_k p_k O_k \delta(x - \tau \lambda_k).
\end{align}
Note that $p_O(x)$ can be negative at some points.

Suppose the eigenvalues of $\tau H$ is within $[-\pi/3, \pi/3]$. Then,  we define the $O$-weighted CDF and ACDF for $p_O(x)$ similar to \cite{lt21}:
\begin{align}
    C_O(x):=(H * p_O)(x), ~~~\wt{C_O}(x) := (F * p_O)(x),
\end{align}
where $H$ is the $2\pi$-periodic Heaviside function and $F=F_{d,\delta}$ is the Fourier approximation of $H$ constructed by Lemma~\ref{lem:approx_Heaviside}. It is easy to verify that $C_O(x)$ equals to $\sum_{k}p_kO_k\mathbf{1}_{x\geq p_kO_k}$ for any $x\in [-\pi/3, \pi/3]$. 

The following lemma gives an unbiased estimator for the $O$-weighted ACDF.

\begin{lemma}[Estimating the $O$-weighted ACDF]\label{lem:est_acdf_observable}
For any $x\in [-\pi, \pi]$, there exists an unbiased estimator $\ov{G_O}(x)$ for the $O$-weighted ACDF $\wt{C_O}(x)$ with variance $\wt{O}(1)$. 

Furthermore, $\ov{G_O}(x)$ runs the quantum circuit (Figure~\ref{fig:hadamard_test_o}) with expected total evolution time $O(\tau d/\log(d))$, where $d$ is the Fourier degree of $F$ .
\end{lemma}
\begin{proof}
$\wt{C_O}(x)$ can be expanded in the following way:
\begin{align}
    \wt{C_O}(x) = &~ (F * p_O)(x)\\
    = &~ \int_{-\pi}^\pi F(x-y) p_O(y) \d y\notag\\
    = &~ \sum_{|j|\leq d} \int_{-\pi}^\pi \hat{F}_j e^{ij(x-y)} p_O(y)\d y\notag\\
    = &~ \sum_{|j|\leq d} \hat{F}_j e^{ijx} \int_{-\pi}^\pi p_O(y) e^{-ijy}\d y\notag\\
    = &~ \sum_{|j|\leq d} \hat{F}_j e^{ijx} \sum_k p_k O_ke^{-ij\tau \lambda_k}\notag\\
    = &~ \sum_{|j|\leq d} \hat{F}_j e^{ijx} \cdot \bra{\phi_0} O e^{-ij\tau H} \ket{\phi_0},
\end{align}
where the third step follows from the Fourier expansion of $F(x-y)$, the fifth step follows from the property of Dirac's delta function, and the last step follows from the definition of $p_k$ and the eigenvalues of matrix exponential.

Define an estimator $G(x; {\bf J}, {\bf Z})$ as follows:
\begin{align}
    G(x; {\bf J}, {\bf Z}):={\cal F}\cdot {\bf Z} e^{i(\theta_{\bf J} + {\bf J}x)},
\end{align}
where $\theta_j$ is defined by $\hat{F}_j = |\hat{F}_j|e^{i\theta_j}$, ${\bf Z}=X+iY$ measured from the quantum circuit (Figure~\ref{fig:hadamard_test_o}) with parameter $j={\bf J}$, and ${\cal F}=\sum_{|j|\leq d}|\hat{F}_j|$. 

Then, we show that $G(x; {\bf J}, {\bf Z})$ is un-biased:
\begin{align}
    \E[G(x; {\bf J}, {\bf Z})] = &~ \sum_{|j|\leq d} \E\left[(X_j + iY_j)e^{i(\theta_j + jx)}|\hat{F}_j|\right]\\
    = &~ \sum_{|j|\leq d} \hat{F}_j e^{ijx} \cdot \E\left[X_j + iY_j\right]\notag\\
    = &~ \sum_{|j|\leq d} \hat{F}_j e^{ijx} \cdot \bra{\phi_0} O e^{-ij\tau H} \ket{\phi_0}\notag\\
    = &~ \wt{C}(x),
\end{align}
where the third step follows from Claim~\ref{clm:estimator_expectation_observable}. Moreover, the variance of $G$ can be upper-bounded by:
\begin{align}
    \Var[G(x; {\bf J}, {\bf Z})]= &~ \E[|G(x; {\bf J}, {\bf Z})|^2] - |\E[G(x; {\bf J}, {\bf Z})]|^2\\
    \leq &~ \E[|G(x; {\bf J}, {\bf Z})|^2]\notag\\
    \leq &~ 2{\cal F}^2,
\end{align}
where the third step follows from $|e^{i(\theta_J+Jx)}|=1$, and the last step follows from $X_j, Y_j\in \{\pm 1\}$.
By Lemma~\ref{lem:approx_Heaviside}, we know that $|\hat{F}_j|=O(1/|j|)$. Hence,
we have ${\cal F} = \sum_{|j|\leq d}O(1/|j|) = O(\log d)$. Thus, $\Var[G(x; {\bf J}, {\bf Z})]=O(\log^2(d))$.

The expected total evolution time is
\begin{align}
    {\cal T}_{\mathsf{tot}} := \E[|J|]=  \tau \sum_{|j|\leq d}|j|\cdot \frac{|\hat{F}_j| }{{\cal F}}= O(\tau d / \log(d)).
\end{align}

The lemma is then proved.
\end{proof}

The following lemma shows that the $O$-weighted CDF $C_O(x)$ can be approximated by the $O$-weighted ACDF $\wt{C_O}(x)$:
\begin{lemma}[Approximating the $O$-weighted CDF]\label{lem:approx_o_cdf}
For any $\epsilon>0$, $0<\delta < \pi/6$, let $F(x) := F_{d,\delta}(x)$ constructed by Lemma~\ref{lem:approx_Heaviside} with approximation error $\eta \epsilon/8$. Then, for any $x\in [-\pi/3, \pi/3]$, it holds that:
\begin{align}
    C_O(x-\delta) - \eta\epsilon/8 \leq \wt{C_O}(x) \leq C_O(x + \delta) + \eta \epsilon/8.
\end{align}
\end{lemma}
The proof is very similar to Lemma~\ref{lem:approx_acdf}, so we omit it here.

We can take $\delta:= \tau \gamma/5$ and let $x_{\mathsf{good}}:=x^\star + \tau \gamma/2$. Then, by Claim~\ref{clm:good_x}, we know that $x_{\mathsf{good}}$ is good for $\lambda_0$, i.e., $[x_{\mathsf{good}}-\tau \gamma, x_{\mathsf{good}} + \tau\gamma]\subset (\tau \lambda_0, \tau \lambda_1)$. Hence, $\wt{C}_O(x_{\mathsf{good}})$  satisfies
\begin{align}
    \left|\wt{C}_O(x_{\mathsf{good}})-p_0 O_0\right|\leq \eta \epsilon /8.
\end{align}

The following lemma shows how to estimate $\wt{C_O}(x_{\mathsf{good}})$, which is very similar to Lemma~\ref{lem:est_overlap}.
\begin{lemma}[Estimating $p_0O_0$]\label{lem:est_p0O0}
For any $\epsilon_1, \nu\in (0, 1)$, $p_0O_0$ can be estimated with multiplicative error $1\pm O(\epsilon_1)$ with probability $1-\nu$ by runs the quantum circuit (Figure~\ref{fig:hadamard_test}) $\wt{O}(\epsilon_1^{-2}\eta^{-2})$ times with expected total evolution time $\wt{O}(\gamma^{-1} \epsilon_1^{-2}\eta^{-2})$ and maximal evolution time $O(\gamma^{-1})$.
\end{lemma}

\subsection{Putting it all together}
In this section, we will put the components together and prove the following main theorem, which gives an algorithm for the ground state property estimation.

\begin{theorem}[Ground state property estimation with commutative observable, restate]\label{thm:app_sim_com}
Suppose $p_0\geq \eta$ for some known $\eta$, and let $\gamma>0$ be the spectral gap of the Hamiltonian. Then, for any $\epsilon, \nu\in (0, 1)$, the ground state property $\bra{\psi_0}O\ket{\psi_0}$ can be estimated within additive error at most $\epsilon$ with probability $1-\nu$, such that:
\begin{enumerate}
    \item the number of times running the quantum circuits (Figure~\ref{fig:hadamard_test} and \ref{fig:hadamard_test_o}) is $\wt{O}(\epsilon^{-2}\eta^{-2})$,
    \item the expected total evolution time is $\wt{O}(\gamma^{-1}\epsilon^{-2}\eta^{-2})$,
    \item the maximal evolution time is $\wt{O}(\gamma^{-1})$.
\end{enumerate}
\end{theorem}
\begin{proof}
By Lemma~\ref{lem:est_overlap}, we obtain an estimate $\ov{p_0}$ for $p_0$ with the guarantee that
\begin{align}\label{eq:approx_p_0}
    \big|\ov{p_0}-p_0\big|\leq O(\eta \epsilon_0),
\end{align}
where $\epsilon_0$ will be chosen shortly.

By Lemma~\ref{lem:est_p0O0}, we obtain an estimate $\ov{p_0O_0}$ for $p_0O_0$ with the guarantee that 
\begin{align}\label{eq:approx_p0O0}
    \big|\ov{p_0O_0}-p_0O_0\big|\leq O(\eta\epsilon_1), 
\end{align}
where $\epsilon_1$ will be chosen shortly.

Then, we have
\begin{align}
    \left|\frac{\ov{p_0O_0}}{\ov{p_0}}-O_0\right|= &~ \left|\frac{\ov{p_0O_0}}{\ov{p_0}} - \frac{p_0O_0}{\ov{p_0}}+\frac{p_0O_0}{\ov{p_0}}-\frac{p_0O_0}{p_0}\right|\\
    \leq &~ \frac{|\ov{p_0O_0}-p_0O_0|}{\ov{p_0}} + |p_0O_0|\left|\frac{1}{\ov{p_0}}-\frac{1}{p_0}\right|\notag\\
    \leq &~ \frac{O(\eta\epsilon_1)}{p_0-O(\eta\epsilon_0)} + |p_0O_0|\left|\frac{1}{p_0-O(\eta\epsilon_0)}-\frac{1}{p_0}\right|\notag\\
    \leq &~ \frac{O(\eta\epsilon_1)}{\eta-O(\eta\epsilon_0)}+|p_0O_0|\left|\frac{1}{p_0-p_0O(\epsilon_0)}-\frac{1}{p_0}\right|\notag\\
    \leq &~ O(\epsilon_1)(1-O(\epsilon_0)) + |O_0|(1+O(\epsilon_0)-1)\notag\\
    \leq &~ O(\epsilon_0+\epsilon_1),
\end{align}
where the second step follows from the triangle inequality, the third step follows from Eqs.~\eqref{eq:approx_p_0} and \eqref{eq:approx_p0O0}, the third step follows from $p_0\geq \eta$, the fifth step follows from $\frac{1}{1-x}\leq 1+O(x)$ for $x\in (0,1)$.

Hence, if we take $\epsilon_0=\epsilon_1=O(\epsilon)$, we will achieve additive error at most $\epsilon$.

For the success probability, we can make Eq.\eqref{eq:approx_p_0} hold with probability $1-\nu/2$ in Lemma~\ref{lem:est_overlap} and  Eq.\eqref{eq:approx_p0O0} hold with probability $1-\nu/2$ in Lemma~\ref{lem:est_p0O0}. Then, by the union bound, we get a good estimate with probability at least $1-\nu$.

The computation costs follow directly from Lemma~\ref{lem:est_overlap} and Lemma~\ref{lem:est_p0O0}. And the proof of the theorem is then completed.
\end{proof}

\begin{algorithm}[ht]
\caption{Ground State Property Estimation (Commutative Case)}
\label{alg:gs_prop_com} 
\begin{algorithmic}[1]
    \algrenewcommand\algorithmicprocedure{\textbf{procedure}}
	\Procedure{EstimateGSProp}{$\epsilon,\tau, \eta, \gamma, \nu$}

        \State $\delta \gets O(\tau \gamma)$, $d\gets O(\delta^{-1}\log(\delta^{-1}\epsilon^{-1}\eta^{-1}))$
        \For{$j\gets -d,\dots,d$}
            \State Compute $\hat{F}_j:=\hat{F}_{d,\delta,j}$ and $\theta_j$ 
        \EndFor
        \State \Comment{Estimate the ground state energy}
        \State $x^\star\gets \textsc{EstimateGSE}(\gamma/8, \tau, \eta, \nu/10)$
        \State $x_{\mathsf{good}}\gets x^\star + \tau \gamma / 2$
        
	    \State \Comment{Generate samples from the Hadamard test circuits}
        \State $N_g\gets O(\log(1/\nu))$, $K\gets O(\epsilon^{-2})$
	    \For{$k\gets 1,\dots,N_g K$}
	        \State Sample $(Z_k, J_k)$ from the quantum circuit (Figure~\ref{fig:hadamard_test})
	        \State Sample $(Z_k', J_k')$ from the quantum circuit (Figure~\ref{fig:hadamard_test_o})
	    \EndFor
	    \State \Comment{Estimate $p_0$}
	    \For{$i\gets 1,\dots, N_g$}
	        \State $\ov{G}_i\gets \frac{1}{K}\sum_{j=1}^K\ov{G}(x_{\mathsf{good}}; Z_{(i-1)K+j}, J_{(i-1)K+j})$
	    \EndFor
	    \State $\ov{p_0}\gets \mathrm{median}(\ov{G}_1,\dots,\ov{G}_{N_g})$
	    \State \Comment{Estimate $p_0O_0$}
	    \For{$i\gets 1,\dots, N_g$}
	        \State $\ov{G}_i'\gets \frac{1}{K}\sum_{j=1}^K\ov{G}(x_{\mathsf{good}}; Z_{(i-1)K+j}', J_{(i-1)K+j}')$
	    \EndFor
	    \State $\ov{p_0O_0}\gets \mathrm{median}(\ov{G}_1',\dots,\ov{G}_{N_g}')$
        \State \Return $\ov{p_0O_0}/\ov{p_0}$
	\EndProcedure
\end{algorithmic}
\end{algorithm}
\section{Algorithm for General Unitary Observables}
\label{sec:unitary_alg}
In this section, we will prove the following theorem for unitary observables in the general case:
\begin{theorem}[Ground state property estimation with general unitary  observable]\label{thm:app_sim}
Suppose $p_0\geq \eta$ for some known $\eta$ and the spectral gap of the Hamiltonian $H$ is at least $\gamma$. For any $\epsilon,\nu\in (0, 1)$, there exists an algorithm for estimating the ground state property $\bra{\psi_0}O\ket{\psi_0}$ within additive error at most $\epsilon$ with probability at least $1-\nu$, such that:
\begin{enumerate}
    \item the expected total evolution time is $\wt{O}(\gamma^{-1}\epsilon^{-2}\eta^{-2})$
    \item the maximal evolution time is $\wt{O}(\gamma^{-1})$.
\end{enumerate}
\end{theorem}

In the following parts, we will first introduce the 2-d $O$-weighted density function and CDF, which extend the commuting observables to the general case. Then, we will show how to combine them with the overlap estimation in Section~\ref{sec:est_overlap} for proving Theorem~\ref{thm:app_sim}.
\subsection{2-d $O$-weighted density function and CDF}
Let $\ket{\phi_0} = \sum_k c_k \ket{\psi_k}$ where $|c_k|^2 = p_k$. In general, $O$ and $H$ may not commute. Hence, we consider a more symmetric form of expectation: $\bra{\phi_0}e^{-ij\tau H} Oe^{-ij'\tau H}\ket{\phi_0}$, which can be expanded in the eigenbasis of $H$ as follows: 
\begin{align}
    \bra{\phi_0}e^{-ij\tau H} O e^{-ij'\tau H}\ket{\phi_0} = &~ \sum_{k,k'} c_k^*  c_{k'} e^{-ij\tau \lambda_k} e^{-ij'\tau \lambda_{k'}}\bra{\psi_k} O \ket{\psi_k'}\notag\\
    =&~ \sum_{k,k'} c_k^* c_{k'} e^{-ij\tau \lambda_{k'}}e^{-ij'\tau \lambda_{k'}}\bra{\psi_k}O\ket{\psi_{k'}}
\end{align}

Similar to the commutative case, we define a 2-d $O$-weighted density function:
\begin{align}
    p_{O,2}(x,y) := \sum_{k,k'}c_k^*c_{k'} O_{k,k'}\delta(x-\tau\lambda_k) \delta(y-\tau\lambda_{k'}),
\end{align}
where $O_{k,k'}:=\bra{\psi_k}O\ket{\psi_{k'}}$. Then, define the corresponding 2-d $O$-weighted CDF function as follows:
\begin{align}
    C_{O,2}(x):=(H_2 * p_{O,2})(x,y),
\end{align}
where $H_2(x,y):=H(x)\cdot H(y)$, the 2-d $2\pi$-periodic Heaviside function. 

We first justify that $C_{O,2}$ is indeed a CDF of $p_{O,2}$ in $[-\pi/3, \pi/3]$:
\begin{align}
    C_2(x,y) = &~ \int_{-\pi}^\pi \int_{-\pi}^\pi H_2(x - u, y-v) p(u,v) \d u \d v\\
    = &~ \sum_{k,k'} c_k^*c_{k'} O_{k,k'}\cdot \int_{-\pi}^\pi \int_{-\pi}^\pi H_2(x- u,y-v) \delta(u-\tau\lambda_k) \delta(v-\tau\lambda_{k'})\d u\d v\notag\\
    = &~ \sum_{k,k'} c_k^*c_{k'} O_{k,k'}\cdot H(x - \tau \lambda_k)H(y - \tau \lambda_{k'})\notag\\
    = &~ \sum_{k,k'} c_k^*c_{k'} O_{k,k'} \cdot \mathbf{1}_{x \geq \tau \lambda_k,y\geq \tau\lambda_{k'}}\notag\\
    = &~ \sum_{\substack{k: \tau\lambda_k \leq x,\\k': \tau \lambda_{k'}\leq y}} c_k^*c_{k'} O_{k,k'}.
\end{align}
Hence, the definition of $C_{O,2}$ is reasonable. 

Then, we show that $C_{O,2}$ can be approximated similar to the 1-d case. Let $F_2(x)$ be the 2-d approximated Heaviside function, i.e.,
\begin{align}
    F_2(x,y):=F(x) \cdot F(y).
\end{align}
The 2-d $O$-weighted approximated CDF (ACDF) is defined to be
\begin{align}
    \wt{C_{O,2}}(x,y) := (F_2 * p_{O,2})(x,y).
\end{align}
The following lemma shows that $\wt{C_{O,2}}(x,y)$ is close to $C_{O,2}(x',y')$ for some $(x',y')$ close to $(x,y)$.
\begin{lemma}[Approximation ratio of the 2-d $O$-weighted ACDF]\label{lem:approx_acdf_2d}
For any $\epsilon>0$, $0<\delta < \pi/6$, let $F_2(x,y) := F_{d,\delta}(x) \cdot F_{d,\delta}(y)$ constructed by Lemma~\ref{lem:approx_Heaviside}. Then, for any $x,y\in [-\pi/3, \pi/3]$, the 2-d $O$-weighted ACDF $\wt{C_{O,2}}(x,y) = (F_2*p_{O,2})(x,y)$ satisfies:
\begin{align}
    C_{O,2}(x-\delta, y-\delta) -2\epsilon \leq \wt{C_{O,2}}(x,y) \leq C_{O,2}(x + \delta, y+\delta) + 2\epsilon.
\end{align}
\end{lemma}
\begin{proof}
By (2) in Lemma~\ref{lem:approx_Heaviside}, we have
\begin{align}
    |F(x) - H(x)|\leq \epsilon ~~~\forall x\in [-\pi + \delta, -\delta]\cup [\delta, \pi - \delta],
\end{align}
which implies that for all $x,y\in [-\pi + \delta, -\delta]\cup [\delta, \pi - \delta]$,
\begin{align}
    |F_2(x,y) - H_2(x,y)|\leq &~ |F(x)F(y)-H(x)H(y)|\\
    = &~ |F(x)F(y)- F(x) H(y) + F(x) H(y) - H(x)H(y)|\notag\\
    \leq &~ F(x) |F(y)-H(y)| + H(y)|F(x)-H(x)|\notag\\
    \leq &~ (F(x) + H(y))\epsilon\notag\\
    \leq &~ 2\epsilon,
\end{align}
where the last step follows from $F(x)\in [0,1]$ by (1) in Lemma~\ref{lem:approx_Heaviside}. Furthermore, we also have for $x\in [-\delta, \delta]$, $y\in [-\pi + \delta, -\delta]$,
\begin{align}
    |F_2(x,y)-H_2(x,y)| \leq &~ |F(x)F(y)-H(x)H(y)|\\
    = &~ |F(x)F(y)|\tag{$H(y)=0$}\\
    \leq &~ F(y)\notag\\
    \leq &~ \epsilon.
\end{align}
Similarly, for $x\in [-\pi+\delta, -\delta]$, $y\in [-\delta, \delta]$,
\begin{align}
    |F_2(x,y)-H_2(x,y)| \leq \epsilon.
\end{align}

Define $F_{L,2} := F_2(x - \delta,y-\delta)$ such that
\begin{align}\label{eq:F_L}
    |F_{L,2}(x) - H_2(x)|\leq 2\epsilon ~~~\forall (x,y)\in &~  [-\pi + 2\delta, 0]\times [-\pi+2\delta, \pi]\\
    \cup &~ [-\pi+2\delta, \pi]\times [-\pi+2\delta, 0]\notag\\
    \cup &~ [2\delta, \pi]\times [2\delta, \pi].\notag
\end{align}
For $\wt{C_{L,2}}(x,y) := (F_{L,2} * p_{O,2})(x,y)$, we have $\wt{C_{L,2}}(x,y) = \wt{C_{O,2}}(x - \delta,y-\delta)$.

Let $p_2:=p_{O,2}$. Then, for any $x,y \in [-\pi/3, \pi/3]$, we have
\begin{align}\label{eq:approx_C_2}
    &\left|C_{O,2}(x,y) - \wt{C_{L,2}}(x,y)\right| =~ \left|\int_{-\pi}^{\pi}\int_{-\pi}^{\pi} p_2(x-u,y-v) (H_2(u,v) - F_{L,2}(u,v)) \d u \d v\right|\\
    \leq &~ \int_{-\pi}^{\pi}\int_{-\pi}^{\pi} p_2(x-u,y-v) |H_2(u,v) - F_{L,2}(u,v)| \d u \d v\notag\\
    = &~ \left(\int_{-\pi}^0\int_{-\pi}^{\pi} + \int_{0}^{\pi}\int_{-\pi}^0+ \int_{2\delta}^\pi\int_{2\delta}^\pi\right) p_2(x-u,y-v)) |H_2(u,v) - F_{L,2}(u,v)| \d u \d v\notag\\
    & + \left(\int_0^{2\delta}\int_0^{\pi} + \int_0^{\pi}\int_0^{2\delta} - \int_0^{2\delta}\int_{0}^{2\delta}\right) p_2(x-u,y-v)) |H_2(u,v) - F_{L,2}(u,v)| \d u \d v\notag\\
    \leq &~ 2\epsilon\cdot \left(\int_{-\pi}^0\int_{-\pi}^{\pi} + \int_{0}^{\pi}\int_{-\pi}^0+ \int_{2\delta}^\pi\int_{2\delta}^\pi\right) p_2(x-u,y-v)\d u\d v\notag\\
    & + \left(\int_0^{2\delta}\int_0^{\pi} + \int_0^{\pi}\int_0^{2\delta} - \int_0^{2\delta}\int_{0}^{2\delta}\right) p_2(x-u,y-v)) |H_2(u,v) - F_{L,2}(u,v)| \d u \d v\notag\\
    \leq &~ 2\epsilon + \left(\int_0^{2\delta}\int_0^{\pi} + \int_0^{\pi}\int_0^{2\delta} - \int_0^{2\delta}\int_{0}^{2\delta}\right) p_2(x-u,y-v)) |H_2(u,v) - F_{L,2}(u,v)| \d u \d v\notag\\
    \leq &~ 2\epsilon + \left(\int_0^{2\delta}\int_0^{\pi} + \int_0^{\pi}\int_0^{2\delta} - \int_0^{2\delta}\int_{0}^{2\delta}\right) p_2(x-u,y-v) \d u\d v\notag\\
    = &~ 2\epsilon + \left(\int_{x-2\delta}^{x}\int_{y-\pi}^y + \int_{x-\pi}^x\int_{y-2\delta}^y -\int_{x-2\delta}^x\int_{y-2\delta}^y\right) p_2(u,v) \d u\d v\label{eq:integral_last}\\
    = &~ 2\epsilon + C_{O,2}(x,y) - C_{O,2}(x-2\delta,y-2\delta),\notag
\end{align}
where the second step follows from Cauchy-Schwarz inequality, the third step follows from partitioning the integration region, the forth step follows from Eq.~\eqref{eq:F_L} and the fact that $p(x,y)$ is supported in $[-\pi/3,\pi/3]\times [-\pi/3,\pi/3]$ and $\delta<\pi/6$ (see Figure~\ref{fig:integral} (a)), the fifth step follows from $p_{O,2}(x)$ is a density function, the last step follows from $C_{O,2}(x)$ is the CDF of $p_{O,2}(x)$ in $[-\pi, \pi]\times [-\pi,\pi]$ and $x,y\in [-\pi/3,\pi/3]$ (see Figure~\ref{fig:integral} (b)).

\begin{figure}[ht]
    \centering
    \subfigure[] {\includegraphics{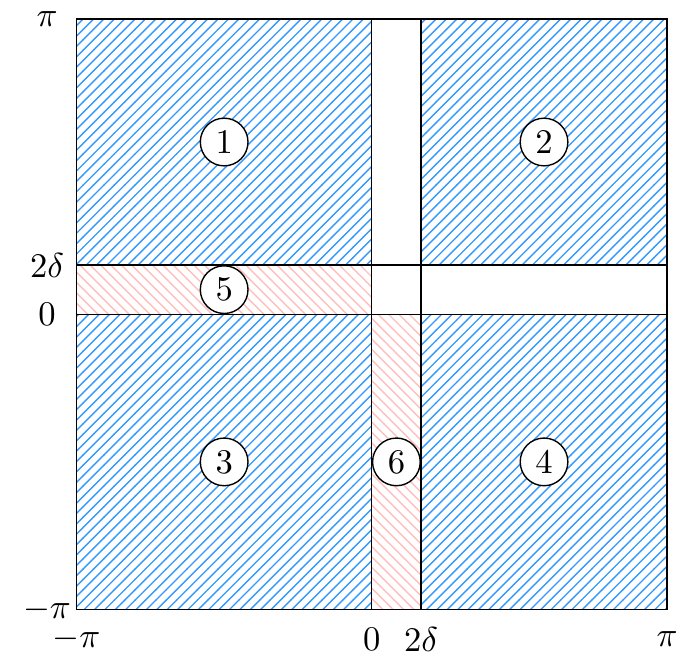}}
    \subfigure[] {\includegraphics[scale=0.95]{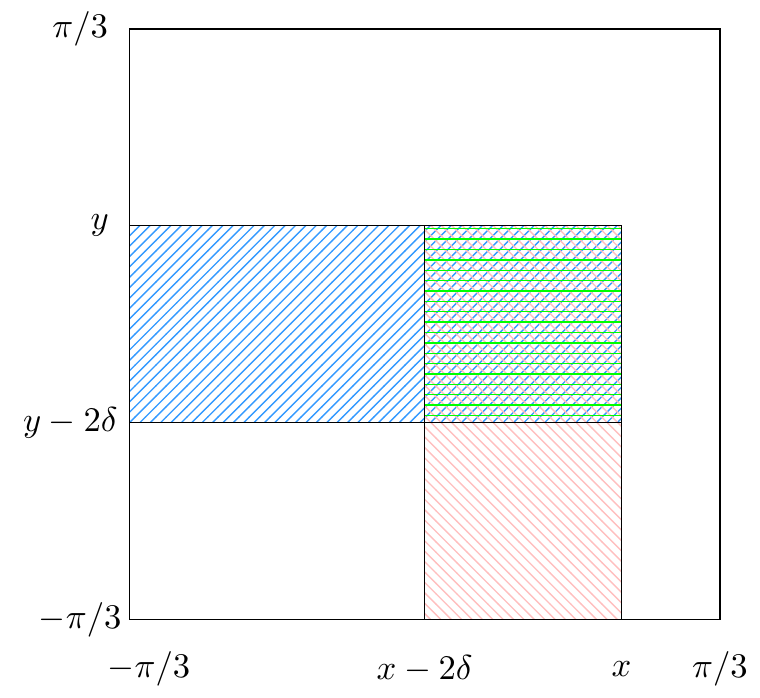}}\label{fig:integral_2}
    \caption{(a) is the integral region for Eq.~\eqref{eq:approx_C_2}, where the integral in regions 1-6 can be upper bounded by Eq.~\eqref{eq:F_L}. (b) is the integral region for Eq.~\eqref{eq:integral_last}.}
    \label{fig:integral}
\end{figure}

Hence, we have
\begin{align}
    \wt{C_{L,2}}(x,y) \geq &~ C_{O,2}(x,y) - (2\epsilon + C_{O,2}(x,y) - C_{O,2}(x - 2\delta,y-2\delta))\notag\\
    = &~ C_{O,2}(x - 2\delta,y-2\delta) - 2\epsilon,
\end{align}
which proves the first inequality: 
\begin{align}
    \wt{C_{O,2}}(x - \delta,y-\delta) \geq C_{O,2}(x-2\delta,y-2\delta) - 2\epsilon.
\end{align}

Similarly, we can define $F_{R,2} := F_2(x + \delta,y+\delta)$ and $\wt{C_{R,2}}(x,y) := (F_{R,2} * p_2)(x,y)$. We can show that
\begin{align}
    \left|C_{O,2}(x,y) - \wt{C_{R,2}}(x,y)\right| \leq 2\epsilon + C_{O,2}(x+2\delta,y+2\delta) - C_{O,2}(x,y),
\end{align}
which gives
\begin{align}
    \wt{C_{O,2}}(x + \delta,y+\delta) \leq C_{O,2}(x + 2\delta,y+2\delta) +2\epsilon.
\end{align}

The lemma is then proved.
\end{proof}

\subsection{Estimating the 2-d ACDF}

We use the following parameterized quantum circuit to estimate the 2-d $O$-weighted ACDF $\wt{C_{O,2}}(x,y)$.
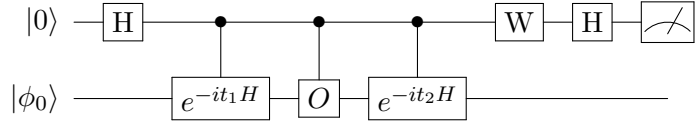
\begin{figure}[H]
    \centering
		\begin{displaymath}
		\Qcircuit @C=1.0em @R=1.2em {
			& & & &\\			
			\lstick{\ket{0}}
			&\gate{\mathrm{H}}	 &\ctrl{1}	& \ctrl{1} & \ctrl{1} & \gate{\mathrm{W}}
			& \gate{\mathrm{H}}			&\meter\\
			\lstick{\ket{\phi_0}} 	 & \qw & \gate{e^{-it_1 H}} & \gate{O} & \gate{e^{-it_2H}}		 
			&\qw &\qw &\qw
 		}		
		\end{displaymath}		
    \caption{Quantum circuit parameterized by $t_1,t_2$. $\mathrm{H}$ is the Hadamard gate and $\mathrm{W}$ is either $I$ or a phase gate $S$. }
    \label{fig:hadamard_test_o_2d}
\end{figure}

\begin{lemma}[Estimate 2-d $O$-weighted ACDF]\label{lem:est_2d_acdf}
For any $x,y\in [-\pi/3, \pi/3]$, for any $\epsilon, \nu\in (0,1)$, we can estimate $\wt{C_{O,2}}(x,y)$ with additive error $\eta\epsilon$ with probability $1-\nu$ by running the quantum circuit (Figure~\ref{fig:hadamard_test_o_2d}) $O(\epsilon^{-2}\eta^{-2}\log(1/\nu))$ times with maximal evolution time $\wt{O}(\gamma^{-1})$ and total expected evolution time $\wt{O}(\gamma^{-1}\epsilon^{-1}\eta^{-1})$.
\end{lemma}

\begin{proof}
$\wt{C_{O,2}}(x,y)$ can be expanded in the following way:
\begin{align}
    \wt{C_{O,2}}(x,y) = &~ (F_2 * p_2)(x,y)\\
    = &~ \int_{-\pi}^\pi \int_{-\pi}^\pi F_2(x-u,y-v) p_2(u,v) \d u \d v\notag\\
    = &~ \sum_{|j|\leq d,|j'|\leq d} \int_{-\pi}^\pi\int_{-\pi}^\pi \hat{F}_j \hat{F}_{j'}e^{ij(x-u)} e^{ij'(y-v)} p_2(u,v)\d u\d v\notag\\
    = &~ \sum_{|j|\leq d,|j'|\leq d} \hat{F}_j\hat{F}_{j'}e^{i(jx+j'y)}\int_{-\pi}^\pi\int_{-\pi}^\pi p_2(u,v) e^{-iju}e^{-ij'v}\d u \d v\notag\\
    = &~ \sum_{|j|\leq d,|j'|\leq d} \hat{F}_j\hat{F}_{j'}e^{i(jx+j'y)} \sum_{k,k'} c_k^*c_k O_{k,k'} e^{-ij\tau\lambda_k}e^{-ij'\tau\lambda_{k'}}\notag\\
    = &~ \sum_{|j|\leq d,|j'|\leq d} \hat{F}_j\hat{F}_{j'}e^{i(jx+j'y)} \cdot \bra{\phi_0} e^{-ij\tau H} O e^{-ij'\tau H}\ket{\phi_0},
\end{align}

To estimate $\bra{\phi_0} e^{-ij\tau H} O e^{-ij'\tau H}\ket{\phi_0}$, we use the multi-level Monte Carlo method. Define a random variables $J,J'$ with support $\{-d, \cdots, d\}$ such that
\begin{align}\label{eq:def_J2}
    \Pr[J=j,J'=j']=\frac{|\hat{F}_j| |\hat{F}_{j'}|}{{\cal F}^2},
\end{align}
where ${\cal F}:=\sum_{|j|\leq d}|\hat{F}_j|$. Then, let $Z:=X_{J,J'} + i Y_{J,J'}\in \{\pm 1\pm i\}$. Define an estimator $\ov{G_2}(x; J, J', Z)$ as follows:
\begin{align}\label{eq:def_g2}
    \ov{G_2}(x,y; J, Z):={\cal F}^2\cdot Z e^{i(\theta_J + Jx)}e^{i(\theta_{J'} + J'y)},
\end{align}
where $\theta_j$ is defined by $\hat{F}_j = |\hat{F}_j|e^{i\theta_j}$, and similar definition for $\theta_{j'}$. Then, we show that $\ov{G_2}(x,y; J, Z)$ is un-biased:
\begin{align}
    \E[\ov{G_2}(x,y; J, J', Z)] = &~ \sum_{|j|\leq d, |j'|\leq d} \E\left[(X_{j,j'} + iY_{j,j'})e^{i(\theta_j + jx)}e^{i(\theta_{j'} + j'y)}|\hat{F}_j||\hat{F}_{j'}|\right]\\
    = &~ \sum_{|j|\leq d, |j'|\leq d} \hat{F}_j\hat{F}_{j'} e^{ijx}e^{ij'y} \cdot \E\left[X_{j,j'} + iY_{j,j'}\right]\notag\\
    = &~ \sum_{|j|\leq d, |j'|\leq d} \hat{F}_j\hat{F}_{j'} e^{ijx}e^{ij'y} \cdot \bra{\phi_0} e^{-ij\tau H} O e^{-ij'\tau H}\ket{\phi_0}\notag\\
    = &~ \wt{C_2}(x,y),
\end{align}
where the third step follows from Claim~\ref{clm:estimator_expectation}. Moreover, the variance of $\ov{G_2}$ can be upper-bounded by:
\begin{align}
    \Var[\ov{G_2}(x,y; J, J', Z)]= &~ \E[|\ov{G_2}(x,y; J,J', Z)|^2] - |\E[\ov{G_2}(x,y; J,J', Z)]|^2\\
    \leq &~ \E[|\ov{G_2}(x,y; J,J', Z)|^2]\notag\\
    = &~ {\cal F}^4 \cdot \E[|X_{J,J'} + i Y_{J,J'}|^2]\notag\\
    = &~ 2{\cal F}^4,
\end{align}
where the third step follows from $|e^{i(\theta_J+Jx)}|=|e^{i(\theta_{J'}+J'y)}|=1$, and the last step follows from $X_{j,j'}, Y_{j,j'}\in \{\pm 1\}$.

By Lemma~\ref{lem:approx_Heaviside}, we know that ${\cal F} = \wt{O}(1)$. Hence, we have for all $x,y\in [-\pi/3, \pi/3]$,
\begin{align}
    \E[\ov{G_2}(x,y)]=\wt{C_{O,2}}(x,y),~~~\text{and}~~~\Var[\ov{G_2}(x,y)]=\wt{O}(1).
\end{align}

Then, using median-of-means estimator, we can obtain an $\epsilon$-additive error estimate of $\wt{C_{O,2}}(x,y)$ with probability $1-\nu$ using $O(\epsilon^{-2}\eta^{-2}\log(1/\nu))$ samples. 

The maximal evolution time is $2d=\wt{O}(\gamma^{-1})$. And the expected evolution time for one trial is
\begin{align}
    \tau \sum_{|j|,|j'|\leq d}(j+j')\frac{|\hat{F}_j||\hat{F}_{j'}|}{{\cal F}^2}=2\tau \sum_{|j|\leq d}j\frac{|\hat{F}_j|}{{\cal F}}=O(\tau d / \log(d)).
\end{align}
Hence, the total expected evolution time is $\wt{O}(\gamma^{-1}\epsilon^{-2}\eta^{-2})$.

The lemma is then proved.
\end{proof}

\begin{figure}[ht!]
    \centering
    \subfigure[] {\includegraphics[scale=2]{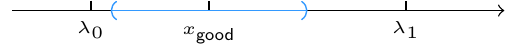}}
    \subfigure[] {\includegraphics[scale=1.5]{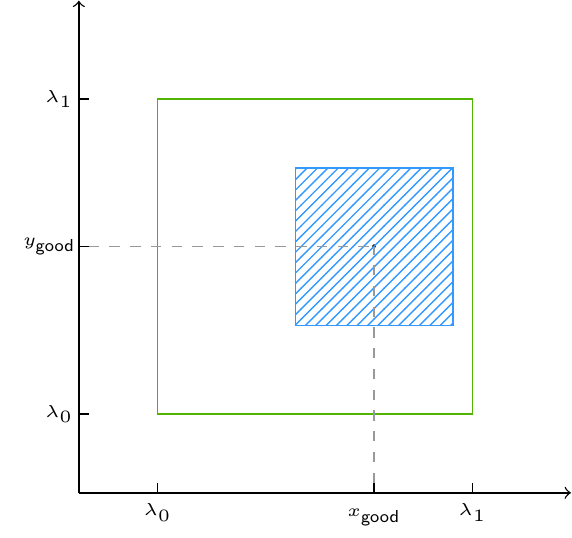}}
    \caption{(a) shows a point that is \emph{good} for $\lambda_0$, where the blue interval is the approximation region such that $\wt{C_O}(x_{\mathsf{good}})$ is close to $C(x)$ for some $x$ in this interval. (b) shows a good point in the 2-d case, where in the green square, the 2-d $O$-weighted CDF $C_{O,2}$ takes the same value $C_{O,2}(\lambda_0, \lambda_0)$. And the blue square is the approximation region of $(x_{\mathsf{good}}, y_{\mathsf{good}})$ such that $\wt{C_{O,2}}(x_{\mathsf{good}}, y_{\mathsf{good}})$ is close to some $C_{O,2}(x,y)$ in this region.}
    \label{fig:my_label}
\end{figure}

Similar to the 1-d case, we can construct a ``good'' point for $(\lambda_0, \lambda_0)$ via the following claim. 
\begin{claim}[Construct 2-d good point]\label{clm:2d_good}
Let $\gamma$ be the spectral gap of the Hamiltonian $H$. Let $x_{\mathsf{good}}:=x^\star + \tau\gamma/2$ where $x^\star$ is the output of $\textsc{EstimateGSE}(\gamma/8, \tau, \eta, \nu/10)$ (Algorithm~\ref{alg:gs_energy}). Then, $(x_{\mathsf{good}}, x_{\mathsf{good}})$ is good for $(\lambda_0, \lambda_0)$. In particular, for any $\epsilon\in (0, 1)$, if the approximation error of $F(x)$ is set to be $\epsilon\eta$, then
\begin{align}
    \left|\wt{C_{O,2}}(x_{\mathsf{good}}, x_{\mathsf{good}})-C_{O,2}(\lambda_0, \lambda_0)\right|\leq 2\epsilon \eta.
\end{align}
\end{claim}
\begin{proof}
By Claim~\ref{clm:good_x}, we know that $x_{\mathsf{good}}$ is good for $\lambda_0$, i.e., $[x_{\mathsf{good}}-\delta, x_{\mathsf{good}} + \delta]$ is contained in $[\lambda_0, \lambda_1)$. It also holds in the 2-d case for $(x_{\mathsf{good}}, x_{\mathsf{good}})$. Then, by Lemma~\ref{lem:est_2d_acdf}, we have
\begin{align}
    C_{O,2}(x_{\mathsf{good}}-\delta, x_{\mathsf{good}}-\delta) -2\epsilon \eta \leq \wt{C_{O,2}}(x_{\mathsf{good}},x_{\mathsf{good}}) \leq C_{O,2}(x_{\mathsf{good}} + \delta, x_{\mathsf{good}}+\delta) + 2\epsilon \eta.
\end{align}
The claim then follows from $C_{O,2}(x,y)=C_{O,2}(\lambda_0, \lambda_0)$ for any $(x,y)\in [\lambda_0, \lambda_1)\times [\lambda_0, \lambda_1)$.
\end{proof}

\subsection{Putting it all together}
The main algorithm for the ground state property estimation will first estimate the ground state energy $\lambda_0$ and the overlap $p_0$, which are described in Section~\ref{sec:est_overlap}. Then, by Lemma~\ref{lem:est_2d_acdf} and Claim~\ref{clm:2d_good}, the weighted expectation $p_0O_0$ can also be estimated. Hence, we will obtain an estimate for $O_0=\bra{\psi_0}O\ket{\psi_0}$.

\begin{algorithm}[ht]
\caption{Ground State Property Estimation (General Case)}
\label{alg:gs_prop} 
\begin{algorithmic}[1]
    \algrenewcommand\algorithmicprocedure{\textbf{procedure}}
	\Procedure{EstimateGSProp}{$\epsilon,\tau, \eta, \gamma, \nu$}

        \State $\delta \gets O(\tau \gamma)$, $d\gets O(\delta^{-1}\log(\delta^{-1}\epsilon^{-1}\eta^{-1}))$
        \For{$j\gets -d,\dots,d$}
            \State Compute $\hat{F}_j:=\hat{F}_{d,\delta,j}$ and $\theta_j$ 
        \EndFor
        \State \Comment{Estimate the ground state energy}
        \State $x^\star\gets \textsc{EstimateGSE}(\gamma/8, \tau, \eta, \nu/10)$
        \State $x_{\mathsf{good}}\gets x^\star + \tau \gamma / 2$
        
	    \State \Comment{Generate samples from the Hadamard test circuits}
        \State $B\gets O(\log(1/\nu))$, $K\gets \wt{O}(\epsilon^{-2})$
	    \For{$k\gets 1,\dots,B K$}
	        \State Sample $(Z_k, J_k)$ from the quantum circuit (Figure~\ref{fig:hadamard_test})
	        \State Sample $(Z_k'', J_{k,1}'', J_{k,2}'')$ from the quantum circuit (Figure~\ref{fig:hadamard_test_o_2d})
	    \EndFor
	    \State \Comment{Estimate $p_0$}
	    \For{$i\gets 1,\dots, B$}
	        \State $\ov{G}_i\gets \frac{1}{K}\sum_{j=1}^K\ov{G}(x_{\mathsf{good}}; Z_{(i-1)K+j}, J_{(i-1)K+j})$
	    \EndFor
	    \State $\ov{p_0}\gets \mathrm{median}(\ov{G}_1,\dots,\ov{G}_{B})$\label{ln:prop_p0}
	    \State \Comment{Estimate $p_0O_0$}
	    \For{$i\gets 1,\dots, B$}
	        \State $\ov{G}_i''\gets \frac{1}{K}\sum_{j=1}^K\ov{G_2}(x_{\mathsf{good}},x_{\mathsf{good}}; Z_{(i-1)K+j}'', J_{(i-1)K+j, 1}'', J_{(i-1)K+j, 2}'')$\Comment{Eq.~\eqref{eq:def_g2}}
	    \EndFor
	    \State $\ov{p_0O_0}\gets \mathrm{median}(\ov{G}_1'',\dots,\ov{G}_{B}'')$\label{ln:prop_p0O0}
        \State \Return $\ov{p_0O_0}/\ov{p_0}$
	\EndProcedure
\end{algorithmic}
\end{algorithm}

\begin{proof}[Proof of Theorem~\ref{thm:app_sim}]
We first analyze the estimation error of Algorithm~\ref{alg:gs_prop}. By Lemma~\ref{lem:est_overlap}, $\ov{p_0}$ (Line~\ref{ln:prop_p0}) has additive error at most $O(\eta\epsilon)$. By Lemma~\ref{lem:est_2d_acdf} and Claim~\ref{clm:2d_good}, $\ov{p_0O_0}$ (Line~\ref{ln:prop_p0O0}) has additive error at most $O(\eta\epsilon)$. Then, by a similar error propagation analysis in Theorem~\ref{thm:app_sim_com}, we get that
\begin{align}
    \left|\frac{\ov{p_0O_0}}{\ov{p_0}}-O_0\right|\leq O(\epsilon).
\end{align}

For the success probability, Algorithm~\ref{alg:gs_prop} has three components: estimate ground state energy, estimate $p_0$, and estimate $p_0O_0$. By our choice of parameters, each of them will fail with probability at most $\nu/3$. Hence, Algorithm~\ref{alg:gs_prop} will succeed with probability at least $1-\nu$.

The maximal evolution time and the total expected evolution time follows from Theorem~\ref{thm:lt21_main}, Lemma~\ref{lem:est_overlap}, and Lemma~\ref{lem:est_2d_acdf}.
\end{proof}

\section{Handling non-unitary observables}
\label{sec:general_alg}
One may notice that Algorithm \ref{alg:gs_prop} works only for unitary observables because it needs to use the circuit in Figure \ref{fig:hadamard_test_o_2d} to estimate $\bra{\phi_0}e^{-it_2 H} Oe^{-it_1 H}\ket{\phi_0}$ for certain $t_1, t_2 \in \mathbb{R}$, in which controlled-$O$ must be a unitary operation. In this section, we show that under reasonable assumptions this algorithm can be modified to estimate the ground state property $\bra{\psi_0} O \ket{\psi_0}$ where $O$ is a general observable. 

Before we present this result, one may wonder why it is necessary. After all, we can always decompose $O$ into a linear combination of Pauli strings $O=\sum_{\vec s} w_{\vec s} P_{\vec s}$, and use Algorithm \ref{alg:gs_prop} to estimate each term $\mu_{\vec s} := \bra{\psi_0} P_{\vec s} \ket{\psi_0}$ individually, and return $\sum_{\vec s} w_{\vec s} \mu_{\vec s}$ as the result. While this strategy works in principle, it might be not efficient enough to be practical, depending on the weights $w_{\vec s}$'s of Pauli strings in the linear expansion of $O$.

Alternatively, one can fix the issue of Algorithm \ref{alg:gs_prop} by designing a procedure for estimating $\bra{\phi_0}e^{-it_2 H} Oe^{-i t_1 H}\ket{\phi_0}$ for arbitrary non-unitary $O$. Such quantities are utilized in the same way as before. We have followed this approach and found that it is possible when there is a block-encoding of $O$. Namely, suppose $O$ is an $n$-qubit observable with $\|O\| \le 1$ and $U$ is an $(n+m)$-qubit unitary operator such that 
\begin{align}
(\bra{0^m}\otimes I) U (\ket{0^m} \otimes I) = \alpha^{-1} O 
\end{align}
for some $\alpha \ge \|O\|$. More details about the block-encoding model can be found in \cite{cgj18,lc19,gslw19,ral20}. Then we can still perform Hadamard test for $U$ to estimate $\bra{\phi_0}e^{-it_2 H} Oe^{-it_1 H}\ket{\phi_0}$ for arbitrary $t_1, t_2 \in \mathbb{R}$. The main theorem of this section is stated below:

\begin{theorem}[Ground state property estimation with block-encoded observable]\label{thm:app_sim_block}
Suppose $p_0\geq \eta$ for some known $\eta$ and the spectral gap of the Hamiltonian $H$ is at least $\gamma$. Suppose we have access to the $\alpha$-block-encoding of the observable $O$. For any $\epsilon,\nu\in (0, 1)$, there exists an algorithm for estimating the ground state property $\bra{\psi_0}O\ket{\psi_0}$ within additive error at most $\epsilon$ with probability at least $1-\nu$, such that:
\begin{enumerate}
    \item the expected total evolution time is $\wt{O}(\gamma^{-1}\epsilon^{-2}\eta^{-2}\alpha^2)$,
    \item the maximal evolution time is $\wt{O}(\gamma^{-1})$.
\end{enumerate}
\label{thm:gspe_complexity_general_case}
\end{theorem}

\begin{proof}[Proof sketch of Theorem \ref{thm:app_sim_block}]
The algorithm for handling non-unitary block-encoded observables is quite similar to Algorithm \ref{alg:gs_prop} for handling unitary observables, except that it relies on a different procedure to estimate $\bra{\phi_0}e^{-it_2 H} Oe^{-it_1 H}\ket{\phi_0}$ for arbitrary $t_1, t_2 \in \mathbb{R}$. Here we briefly describe this procedure and defer the detailed analysis to Appendix \ref{sec:hadamard_test_block_encoding}. 

Let $C \mhyphen V := \ket{0}\bra{0} \otimes I + \ket{1}\bra{1} \otimes V$ be the controlled-V operation for arbitrary unitary operator $V$. Let $\ket{\phi_0}$ be an arbitrary $n$-qubit state. 
Consider the following procedure (as illustrated in Figure \ref{fig:hadamard_test_block_encoding}:

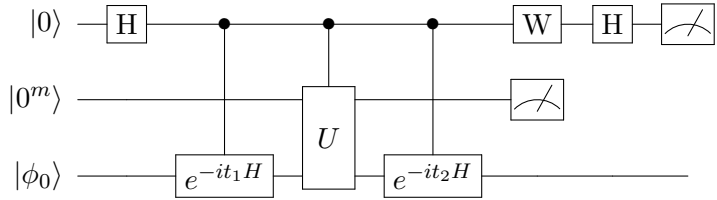
\begin{figure}[ht]
    \centering
		\begin{displaymath}
		\Qcircuit @C=1.0em @R=1.2em {
			& & & &\\			
			\lstick{\ket{0}}
			&\gate{\mathrm{H}}	 &\ctrl{2}	& \ctrl{1} & \ctrl{2} & \gate{\mathrm{W}}
			& \gate{\mathrm{H}}			&\meter\\
			\lstick{\ket{0^m}} 	 & \qw & \qw & \multigate{1}{U} & \qw	  &\meter\\
			\lstick{\ket{\phi_0}} & \qw & \gate{e^{-it_1 H}} & \ghost{U} & \gate{e^{-it_2 H}} & \qw & \qw &\qw
 		}		
		\end{displaymath}		
    \caption{Quantum circuit parameterized by $t_1,t_2$. $\mathrm{H}$ is the Hadamard gate and $\mathrm{W}$ is either $I$ or a phase gate $S$. $U$ is the block-encoding of the non-unitary observable $O$. }
    \label{fig:hadamard_test_block_encoding}
\end{figure}

\begin{enumerate}
    \item Prepare the state $\ket{0} \ket{0^m} \ket{\phi_0}$.
    \item Apply a Hadamard gate on the first register.
    \item Apply a $C \mhyphen e^{-i H t_1}$ on the first and third registers.
    \item Apply $C \mhyphen U$ on the current state, obtaining
    \begin{align}
        \dfrac{1}{\sqrt{2}} \left (  
        \ket{0} \ket{0^m} \ket{\phi_0} + \ket{1} U \ket{0^m} e^{-i H t_1}\ket{\phi_0}
        \right).
    \end{align}    
    \item Measure the second register in the standard basis. If the outcome is not $0^m$, then this procedure fails; otherwise, continue. The probability of this step succeeding is
    \begin{align}
        p_{succ} = \dfrac{1+ \alpha^{-2} \bra{\phi_0} e^{i H t_1} O^2 e^{-i H t_1} \ket{\phi_0}} {2},
    \end{align}
    and when this event happens, the state becomes
    \begin{align}
        \dfrac{1}{\sqrt{2 p_{succ}}} \left [
        \ket{0} \ket{\phi_0} + \alpha^{-1} \ket{1} O e^{-i H t_1} \ket{\phi_0} 
        \right ].
    \end{align}
    \item Apply a $C \mhyphen e^{-i H t_2}$ on the first and third registers. The state becomes
    \begin{align}
        \dfrac{1}{\sqrt{2 p_{succ}}} \left [
        \ket{0} \ket{\phi_0} + \alpha^{-1} \ket{1} e^{-i H t_2} O e^{-i H t_1} \ket{\phi_0} 
        \right ].
    \end{align}    
    \item Apply $W=I$ or phase gate $S$ on the first register.
    \item Apply a Hadamard gate on the first register.
    \item Measure the first register in the standard basis. Then if $W=I$, the (conditional) probability of getting outcome $0$ is 
    \begin{align}
    \mathbb{P}[0 | succ] =  \dfrac{p_{succ} + \alpha^{-1} \operatorname{Re}[ \bra{\phi_0} e^{-i H t_2} O e^{-i H t_1} \ket{\phi_0}] }{2 p_{succ}};
    \end{align}
    if $W=S$, this probability is    
    \begin{align}
    \mathbb{P}[0 | succ] =  \dfrac{p_{succ} - \alpha^{-1} \operatorname{Im}[ \bra{\phi_0} e^{-i H t_2} O e^{-i H t_1} \ket{\phi_0}] }{2 p_{succ}}.
    \end{align}     
\end{enumerate}

Now we define two random variables $X$ and $Y$ as follows. First, we run the above procedure with $W=I$ in step 7. If step 5 fails, $X=0$; otherwise, if the measurement outcome is $0$ or $1$ in step 9, then $X=\alpha$ or $-\alpha$, respectively. One can show that $X$ is an unbiased estimator of $\operatorname{Re}[ \bra{\phi_0} e^{-i H t_2} O e^{-i H t_1} \ket{\phi_0}]$, i.e.
\begin{align}
\mathbb{E}[X]=\operatorname{Re}[\bra{\phi_0} e^{-i H t_2} O e^{-i H t_1} \ket{\phi_0}].
\end{align}
$Y$ is defined similarly. We run the above procedure with $W=S$ in step 7. If step 5 fails, $Y=0$; otherwise, if the measurement outcome is $1$ or $0$ in step 9, then $Y=\alpha$ or $-\alpha$, respectively. Then $Y$ is an unbiased estimator of 
$\operatorname{Im}[\bra{\phi_0} e^{-i H t_2} O e^{-i H t_1} \ket{\phi_0}]$, i.e.
\begin{align}
\mathbb{E}[Y]=\operatorname{Im}[\bra{\phi_0} e^{-i H t_2} O e^{-i H t_1} \ket{\phi_0}].
\end{align}
It follows that $Z:=X+iY$ is an unbiased estimator of $\bra{\phi_0} e^{-i H t_2} O e^{-i H t_1} \ket{\phi_0}$, i.e.
\begin{align}
\mathbb{E}[Z]=\bra{\phi_0} e^{-i H t_2} O e^{-i H t_1} \ket{\phi_0}.
\end{align}
Note that $|Z|^2 = |X|^2 + |Y|^2 \le  2\alpha^2$ with certainty. 

Equipped with the above method for estimating $\bra{\phi_0} e^{-i H t_2} O e^{-i H t_1} \ket{\phi_0}$ for arbitrary $t_1, t_2 \in \mathbb{R}$, we can now use the same strategy as in Lemma \ref{lem:est_2d_acdf} to estimate $\wt{C_{O,2}}(x,y)$. The other components of Algorithm \ref{alg:gs_prop} remain intact. The analysis of this modified algorithm is almost the same as before, except that now we have
\begin{align}
\Var[\ov{G_2}(x,y)]=\wt{O}(\alpha^2).
\end{align} 
As a consequence, compared to Theorem \ref{thm:app_sim}, the total evolution time of this modified algorithm is larger by a factor of $O(\alpha^2)$, while its maximal evolution time is of the same order.
\end{proof}

\section{Applications}
\label{sec:apps}
In this section, we discuss some applications of our ground state property estimation algorithm. To define an application of the ground state property estimation algorithm, we must specify a Hamiltonian of interest $H$ and an observable of interest $O$.
An example application used in quantum chemistry and materials is the Green's function (see, e.g. \cite{tong2021fast}), where $O=a_i(z-(H-E_0)^{-1})a_j^\dagger$.
In the following two sections we describe another example from quantum chemistry and materials as well as an example of a linear algebraic subroutine.

\subsection{Charge density}
The primary application of the technique is the estimation of ground state properties of physical systems.
Here we describe how to compute the charge density of a molecule, which can be used to compute properties like electric dipole moments of a molecule \cite{rice2021quantum}.
From a second-quantized representation of the electronic system (assuming fixed positions of the nuclear positions), the charge density is determined from the one-particle reduced density matrix as,
\begin{align}
    \rho(\vec{r}) = -e \sum_{p,q} D_{p,q} \phi_p^{\*}(\vec{r})\phi_q(\vec{r}),
\end{align}
where $e$ is the electric constant, $D_{p,q}$ is the one-electron reduced density matrix (1RDM) of the ground state, and $\phi_q(\vec{r})$ are the basis wave functions chosen for the second-quantized representation of the electronic system \cite{helgaker2014molecular}.
The 1RDM of the ground state is a matrix of properties of the ground state with each entry defined as
\begin{align}
    D_{p,q} = \bra{\psi_0}a_p^{\dagger}a_q\ket{\psi_0},
\end{align}
where $a_p$ are annihilation operators.
The operators involved in the 1RDM can each be expressed as a linear combination of unitary operators using the Majorana representation 
$a_p=\frac{1}{2}(\gamma_{2p}+i\gamma_{2p+1})$,
where the $\gamma_k$ are hermitian and unitary\footnote{To implement this application on a quantum computer we must represent the unitaries as operations on qubits. For an $n$-electron system, using the Jordan-Wigner or Bravyi-Kitaev transformation \cite{seeley2012bravyi}, each Majorana operator, and products thereof, can be represented as a Pauli string.
},
\begin{align}
    D_{p,q} = \frac{1}{4}\left(\bra{\psi_0}\gamma_{2p}\gamma_{2q}\ket{\psi_0}-i\bra{\psi_0}\gamma_{2p+1}\gamma_{2q}\ket{\psi_0}+i\bra{\psi_0}\gamma_{2p}\gamma_{2q+1}\ket{\psi_0}+\bra{\psi_0}\gamma_{2p+1}\gamma_{2q+1}\ket{\psi_0}\right).
\end{align}
Accordingly, we may use the method of Section \ref{sec:unitary_alg} to estimate each entry of the 1RDM and then obtain the charge density function of the ground state.
As a point of comparison, we could alternatively use the variational quantum eigensolver algorithm to prepare an approximation to the ground state and then directly estimate each of the Pauli expectation values.
However, there is no guarantee on whether a target accuracy for the ground state approximation can be achieved.
Remarkably, the methods introduced in this paper can be used to ensure a target accuracy in the estimation regardless of the quality of ground state approximation, though possibly at the cost of an increase in runtime.


\subsection{Quantum linear system solver}
In the seminal \cite{hhl09} paper, a quantum algorithm is proposed to generate a quantum state approximately proportional to the solution of a linear system of equations. Namely, given a linear system $A \vec{x}=\vec{b}$, the algorithm produces a quantum state close to $\ket{x} := \frac{\sum_j x_j \ket{j}}{\sqrt{\sum_j |x_j|^2}}$, where $x_j$'s are the entries of $\vec{x}=A^{-1}\vec{b}$. In fact, in many cases, we only need to know $\bra{x}M \ket{x}$, where  $M$ is a linear operator. For example, in quantum mechanics, many features of $\ket{x}$ can be extracted in this way, including normalization, moments, etc. One approach to solve this problem is first solving the linear system using any quantum linear system solver \cite{hhl09, cks17, cgj18, gslw19} to obtain the state $\ket{x}$ and then performing the measurement of $M$. However, a shortcoming of this method is that most of the quantum linear system solvers require deep quantum circuits. And hence, the needed quantum resources may not be accessible in the near future. 

Recently, a few quantum algorithms \cite{blrm19, hbr19, syso19} were developed  to solve linear systems of equations by encoding such a system into an effective Hamiltonian
\begin{align}
    H_G:=A^\dagger (I-\ket{b}\bra{b})A,
\end{align}
whose ground state corresponds to the solution vector $\ket{x}$. We can combine 
this idea 
with our ground state property estimation algorithm to get a low-depth algorithm for estimating the properties of linear system solution. More specifically, suppose we can simulate the Hamiltonian $H_G$ for some specified time and we know the normalization factor $\tau$ such that the eigenvalues of $\tau H_G$ are in $[-\pi/3, \pi/3]$. For the operator $M$, we can assume that $M$ can be decomposed into a linear combination of Pauli operators $M=\sum_{\ell=1}^L c_\ell \sigma_\ell$, or we assume that $M$ is given in the block-encoding form. The estimation algorithm has two steps:
\begin{enumerate}
    \item Run a quantum linear system algorithm (e.g. \cite{syso19}, \cite{an2019quantum}, or \cite{lin2020optimal}) with constant precision to prepare an initial state $\ket{\phi_0}$ such that $|\bra{\phi_0}x\rangle|^2$ is $\Omega(1)$.
    \item Using $\ket{\phi_0}$ from step 1 as the initial state, run Algorithm~\ref{alg:gs_prop} to estimate $\bra{x}M\ket{x}$ within $\epsilon$-additive error for any $\epsilon\in (0,1)$.
\end{enumerate}
Step 1 takes $\tilde{O}(\kappa)$ time, where $\kappa$ is the condition number of $A$. 
To analyze the computation cost of the second step, we need a lower-bound on the spectral gap of $H_G$. Since $\bra{x}A^\dagger (I-\ket{b}\bra{b})A\ket{x}=0$, we have $\lambda_0(H_G)=0$. For the second smallest eigenvalue, since $H_G=A^\dagger A - A^\dagger \ket{b}\bra{b}A$, by Weyl's inequality, we have
\begin{align}
    \lambda_1(H_G) \geq &~ \lambda_0(A^\dagger A) - \lambda_1(A^\dagger \ket{b}\bra{b}A)\notag\\
    = &~ \lambda_0(A^\dagger A),
\end{align}
where the second step follows from $A^\dagger \ket{b}\bra{b}A$ is rank-1. Due to the normalization, the smallest (normalized) singular value of $A$ is $\Omega(\kappa^{-1})$.
Hence, we have $\gamma= \Omega(\kappa^{-2})$. 

By Theorem~\ref{thm:app_sim}, the maximal evolution time of the Hamiltonian will be $\wt{O}(\kappa^2)$. To further improve the circuit depth, we may apply the gap amplification technique \cite{sb13, syso19} to quadratically increase the spectral gap of $H_G$. Specifically, consider the following family of Hamiltonians:
\begin{align}\label{eq:lin_hamiltonian}
    \bar{H}'_G(s):=\sigma^+\otimes \bar{A}^\dagger(s) (I-\ket{\bar{b}}\bra{\bar{b}}) + \sigma^-\otimes (I-\ket{\bar{b}}\bra{\bar{b}}) \bar{A}(s),
\end{align}
where $\sigma^{\pm}=(X\pm iY)/2$, $\bar{A}(s):=(1-s)Z \otimes I + s X \otimes A$,  $\ket{\bar{b}}:=\ket{+}\ket{b}$ and $s \in [0, 1]$. Note that these Hamiltonians act on the original system and two ancilla qubits. Then we have
\begin{align}
    (\bar{H}'_G(s))^2=\begin{bmatrix}
    \bar{H}_G(s) & 0\\
    0 & (I-\ket{\bar{b}}\bra{\bar{b}}) \bar{A}(s) \bar{A}^\dagger(s) (I-\ket{\bar{b}}\bra{\bar{b}})
    \end{bmatrix},
\end{align}
where
\begin{align}
\bar{H}_G(s):=\bar{A}^\dagger(s)(I-\ket{\bar{b}}\bra{\bar{b}}) \bar{A}(s).
\end{align}
As shown in \cite{syso19}, the eigenvalues of $\bar{H}'_G(s)$ are
\begin{align}
\left\{0, 0, \pm\sqrt{\lambda_1(s)}, \pm\sqrt{\lambda_2(s)}, \dots\right\},    
\end{align}
where $\lambda_j(s)$'s are the nonzero eigenvalues of $\bar{H}_G(s)$. Furthermore, let $\ket{x(s)}$ be the unique ground state of $\bar{H}_G(s)$. Note that 
$\ket{x(0)}=\ket{-}\ket{b}$ and $\ket{x(1)}=\ket{+}\ket{x}$. Then the ground space of $\bar{H}'_G(s)$ is spanned by $\{\ket{0}\ket{x(s)}, \ket{1}\ket{\bar{b}}\}$. In addition, for $s=1$, one can use Weyl's ineqality to show that $\lambda_1(1) \ge \kappa^{-2}$, which implies that the smallest nonzero eigenvalue of $\bar{H}_G'(1)$ is $\Omega(\kappa^{-1})$, as desired.

We can use the algorithm in \cite{syso19} to prepare a state that has $\Omega(1)$ overlap with $\ket{0}\ket{x(1)}=\ket{0}\ket{+}\ket{x}$ in $\tilde{O}(\kappa)$ time. Specifically, this algorithm starts with the state $\ket{0}\ket{x(0)}=\ket{0}\ket{-}\ket{b}$, performs a sequence of unitary operations of form $e^{-i t_k \bar{H}_G'(s_k)}$ on it, and outputs a state $\epsilon$-close to $\ket{0}\ket{x(1)}$ in $\tilde{O}(\kappa \epsilon^{-1})$ time. Here we set $\epsilon=\Theta(1)$ and the time cost of this procedure is $\tilde{O}(\kappa)$. 

After obtaining a state $\ket{\phi_0}$ that has $\Omega(1)$ overlap with $\ket{0}\ket{+}\ket{x}$, we run Algorithm~\ref{alg:gs_prop} on $\ket{\phi_0}$, $\bar{H}'_G(1)$ and $\tilde{M}:=\ket{0}\bra{0}\otimes \ket{+}\bra{+}\otimes M$ to estimate $\bra{0,+,x}\tilde{M}\ket{0,+,x}=\bra{x}M\ket{x}$. Notice that since we know the ground state energy of $\bar{H}'_G(1)$ is zero, we do not need to first estimate the ground state energy using Algorithm~\ref{alg:gs_energy}. Instead, we directly evaluate the $O$-weighted CDF at zero. 
Therefore, by Theorem~\ref{thm:app_sim_block}, we get the following result:

\begin{corollary}[Quantum linear system solution property estimation]
\label{cor:qlss}
For a linear system $A\vec{x}=\vec{b}$, suppose $A$ has singular values in $[-1, -1/\kappa]\cup [1/\kappa, 1]$ for $\kappa>1$, and the eigenvalues of $\bar{H}'_G(1)$ (Eq.~\eqref{eq:lin_hamiltonian}) are in $[-\pi/3, \pi/3]$. Furthermore, suppose we can implement $e^{-i t \bar{H}'_G(s)}$ (Eq.~\eqref{eq:lin_hamiltonian}) in $\tilde{O}(t)$ time for all $s \in [0, 1]$.

Then, for any linear operator $M$ given by its $\alpha$-block encoding unitary $U_M$, and for any $\epsilon\in (0,1)$, the expectation value $\bra{x}M\ket{x}$ can be estimated with $\epsilon$-additive error with high probability such that:
\begin{itemize}
    \item the depth of each circuit is $\wt{O}(\kappa)$.
    \item the expected total runtime is 
    $\wt{O}(\kappa\epsilon^{-2}\alpha^2)$.    
\end{itemize}
\end{corollary}

For comparison, the algorithm in \cite{syso19} needs $\wt{O}(\kappa\epsilon^{-1})$ circuit depth to obtain a state that is $\epsilon$-close to $\ket{x}$, which is larger than ours. Moreover, to estimate $\bra{x}M\ket{x}$, even with amplitude estimation, it still needs $\Omega(\epsilon^{-1})$ copies of the state to achieve $\epsilon$-additive error. Hence, its total runtime will be $\wt{O}(\kappa\epsilon^{-2})$, nearly matching our result (ignoring the dependence on the $\alpha$ factor).

\section{Discussion and Outlook}\label{sec:discuss}
We have shown a quantum-classical hybrid algorithm for estimating properties of the ground state of a Hamiltonian, such that the quantum circuit depth is relatively small and only poly-logarithmically depends on $\epsilon^{-1}$. Therefore, the algorithm has a significant advantage in high-accuracy estimation, and it is possible to be implemented in early fault-tolerant devices. In practice, our algorithm can solve many important tasks by combining with some initial state preparation methods (e.g., VQE or QAOA). In this paper, we provide two examples, one in quantum chemistry and another in solving linear systems. And we believe more applications will be explored in the future.

Another important direction
is to improve the total evolution time of our algorithm which quadratically depends on $\epsilon^{-1}$. The blowup comes from evaluating the $O$-weighted CDF in high precision and a trade-off between maximal evolution time and total evolution time. However, this does not meet the Heisenberg-limit of linear dependence on $\epsilon^{-1}$ for generic Hamiltonians \cite{aa17}. In our main result (Theorem~\ref{thm:app_sim}), the $\epsilon^{-2}\eta^{-2}$ factor comes from the number of samples needed to reduce the estimator's error to $O(\epsilon\eta)$. Amplitude estimation can be used to reduce this number of samples and the total evolution time.
However, this comes at the cost of significantly increasing the maximal evolution time, which could require large fault-tolerant overheads for reliable implementation.
A strategy to achieve improved performance that is more amenable to early fault tolerant quantum computers is to use recently introduced ``enhanced sampling'' techniques \cite{wang2021minimizing}.
If $\lambda$ characterizes the fidelity decay rate of the circuit as deeper circuits are used, then we would expect to need a maximal evolution time of $O(\lambda^{-1}\gamma^{-1})$ and an total evolution time of $O(\lambda\gamma^{-1}\epsilon^{-2}\eta^{-2})$.
Note that because this approach incorporates the impact of error into the algorithm, the maximal evolution time is of no concern.
Rather than being a cost that needs monitoring, the maximal evolution time is chosen by the algorithm to minimize the total evolution time.
With this, we expect that as the quality of devices is improved, the performance of the algorithm improves proportionally.
We note that a similar approach can also be applied to improve the total evolution time in \cite{lt21} from $\wt{O}(\epsilon^{-1}\eta^{-2})$ to
$\wt{O}(\lambda\epsilon^{-1}\eta^{-2})$.

This work fits into the paradigm of ``beyond the ground state energy'' and studies more general properties of the ground state. Can we go further beyond the ground state? Some prior works have explored the estimation of such kind of properties of Hamiltonian. For example, Brown, Flammia, and Schuch \cite{bfs11} studied the density of states. Jordan, Gosset, and Love \cite{jgl10} focused on the energy of excited states. Gharibian and Sikora \cite{gs15} identified the energy barriers. Watson and Bausch \cite{watson2021complexity} explored detecting phase transitions via order parameters.
In general, for an unknown Hamiltonian, these estimation problems will be hard. An interesting open problem is, given some prior knowledge of the Hamiltonian, can we design efficient or low-depth quantum algorithms for estimating Hamiltonian properties beyond ground state? 
\newpage

\noindent\textbf{Acknowledgements} This work was done while R.Z. was a research intern at Zapata Computing Inc. We thank Yu Tong, Phillip Jensen, Max Radin, 
J\'{e}r\^{o}me Gonthier, Alex Kunitsa, and Aram Harrow for useful input and feedback on this work.

\bibliographystyle{quantum}
\bibliography{ref}

\begin{thebibliography}{10}

\bibitem{cao2018potential}
Yudong Cao, Jhonathan Romero, and Al{\'a}n Aspuru-Guzik.
\newblock ``Potential of quantum computing for drug discovery''.
\newblock \href{https://dx.doi.org/10.1147/JRD.2018.2888987}{IBM Journal of
  Research and Development {\bf 62}, 6--1}~(2018).

\bibitem{cao2019quantum}
Yudong Cao, Jonathan Romero, Jonathan~P Olson, Matthias Degroote, Peter~D
  Johnson, M{\'a}ria Kieferov{\'a}, Ian~D Kivlichan, Tim Menke, Borja
  Peropadre, Nicolas~PD Sawaya, et~al.
\newblock ``Quantum chemistry in the age of quantum computing''.
\newblock \href{https://dx.doi.org/10.1021/acs.chemrev.8b00803}{Chemical
  reviews {\bf 119}, 10856--10915}~(2019).

\bibitem{aspuru2005simulated}
Al{\'a}n Aspuru-Guzik, Anthony~D Dutoi, Peter~J Love, and Martin Head-Gordon.
\newblock ``Simulated quantum computation of molecular energies''.
\newblock \href{https://dx.doi.org/10.1126/science.1113479}{Science {\bf 309},
  1704--1707}~(2005).

\bibitem{peruzzo2014variational}
Alberto Peruzzo, Jarrod McClean, Peter Shadbolt, Man-Hong Yung, Xiao-Qi Zhou,
  Peter~J Love, Al{\'a}n Aspuru-Guzik, and Jeremy~L O’brien.
\newblock ``A variational eigenvalue solver on a photonic quantum processor''.
\newblock \href{https://dx.doi.org/10.1038/ncomms5213}{Nature communications
  {\bf 5}, 1--7}~(2014).

\bibitem{meir1992landauer}
Yigal Meir and Ned~S Wingreen.
\newblock ``Landauer formula for the current through an interacting electron
  region''.
\newblock \href{https://dx.doi.org/10.1103/PhysRevLett.68.2512}{Physical review
  letters {\bf 68}, 2512}~(1992).

\bibitem{jensen2017introduction}
Frank Jensen.
\newblock ``Introduction to computational chemistry''.
\newblock John Wiley \& Sons. ~(2017).

\bibitem{o2019calculating}
Thomas~E O’Brien, Bruno Senjean, Ramiro Sagastizabal, Xavier Bonet-Monroig,
  Alicja Dutkiewicz, Francesco Buda, Leonardo DiCarlo, and Lucas Visscher.
\newblock ``Calculating energy derivatives for quantum chemistry on a quantum
  computer''.
\newblock \href{https://dx.doi.org/10.1038/s41534-019-0213-4}{npj Quantum
  Information {\bf 5}, 1--12}~(2019).

\bibitem{amb14}
Andris Ambainis.
\newblock ``On physical problems that are slightly more difficult than qma''.
\newblock In 2014 IEEE 29th Conference on Computational Complexity (CCC).
\newblock \href{https://dx.doi.org/10.1109/CCC.2014.12}{Pages 32--43}.
\newblock ~(2014).

\bibitem{gy19}
Sevag Gharibian and Justin Yirka.
\newblock ``The complexity of simulating local measurements on quantum
  systems''.
\newblock \href{https://dx.doi.org/10.22331/q-2019-09-30-189}{Quantum {\bf 3},
  189}~(2019).

\bibitem{gpy20}
Sevag Gharibian, Stephen Piddock, and Justin Yirka.
\newblock ``{Oracle Complexity Classes and Local Measurements on Physical
  Hamiltonians}''.
\newblock In Christophe Paul and Markus Bl{\"a}ser, editors, 37th International
  Symposium on Theoretical Aspects of Computer Science (STACS 2020).
\newblock \href{https://dx.doi.org/10.4230/LIPIcs.STACS.2020.20}{Volume 154 of
  Leibniz International Proceedings in Informatics (LIPIcs), pages
  20:1--20:37}.
\newblock Dagstuhl, Germany~(2020). Schloss Dagstuhl--Leibniz-Zentrum f{\"u}r
  Informatik.

\bibitem{pw09}
David Poulin and Pawel Wocjan.
\newblock ``Preparing ground states of quantum many-body systems on a quantum
  computer''.
\newblock \href{https://dx.doi.org/10.1103/PhysRevLett.102.130503}{Physical
  review letters {\bf 102}, 130503}~(2009).

\bibitem{gtc19}
Yimin Ge, Jordi Tura, and J~Ignacio Cirac.
\newblock ``Faster ground state preparation and high-precision ground energy
  estimation with fewer qubits''.
\newblock \href{https://dx.doi.org/10.1063/1.5027484}{Journal of Mathematical
  Physics {\bf 60}, 022202}~(2019).

\bibitem{lt20}
Lin Lin and Yu~Tong.
\newblock ``Near-optimal ground state preparation''.
\newblock \href{https://dx.doi.org/10.22331/q-2020-12-14-372}{Quantum {\bf 4},
  372}~(2020).

\bibitem{mcardle2019digital}
Sam McArdle, Alexander Mayorov, Xiao Shan, Simon Benjamin, and Xiao Yuan.
\newblock ``Digital quantum simulation of molecular vibrations''.
\newblock \href{https://dx.doi.org/10.1039/C9SC01313J}{Chemical science {\bf
  10}, 5725--5735}~(2019).

\bibitem{gonthier2020identifying}
J{\'e}r{\^o}me~F. Gonthier, Maxwell~D. Radin, Corneliu Buda, Eric~J. Doskocil,
  Clena~M. Abuan, and Jhonathan Romero.
\newblock ``Identifying challenges towards practical quantum advantage through
  resource estimation: the measurement roadblock in the variational quantum
  eigensolver''~(2020).
\newblock  \href{http://arxiv.org/abs/2012.04001}{arXiv:2012.04001}.

\bibitem{wang2021minimizing}
Guoming Wang, Dax~Enshan Koh, Peter~D Johnson, and Yudong Cao.
\newblock ``Minimizing estimation runtime on noisy quantum computers''.
\newblock \href{https://dx.doi.org/10.1103/PRXQuantum.2.010346}{PRX Quantum
  {\bf 2}, 010346}~(2021).

\bibitem{bmn21}
Ryan Babbush, Jarrod~R McClean, Michael Newman, Craig Gidney, Sergio Boixo, and
  Hartmut Neven.
\newblock ``Focus beyond quadratic speedups for error-corrected quantum
  advantage''.
\newblock \href{https://dx.doi.org/10.1103/PRXQuantum.2.010103}{PRX Quantum
  {\bf 2}, 010103}~(2021).

\bibitem{bom21}
Kyle~EC Booth, Bryan O'Gorman, Jeffrey Marshall, Stuart Hadfield, and Eleanor
  Rieffel.
\newblock ``Quantum-accelerated constraint programming''.
\newblock \href{https://dx.doi.org/10.22331/q-2021-09-28-550}{Quantum {\bf 5},
  550}~(2021).

\bibitem{cam21}
Earl~T Campbell.
\newblock ``Early fault-tolerant simulations of the hubbard model''.
\newblock \href{https://dx.doi.org/10.1088/2058-9565/ac3110}{Quantum Science
  and Technology {\bf 7}, 015007}~(2021).

\bibitem{lt21}
Lin Lin and Yu~Tong.
\newblock ``Heisenberg-limited ground-state energy estimation for early
  fault-tolerant quantum computers''.
\newblock \href{https://dx.doi.org/10.1103/PRXQuantum.3.010318}{PRX Quantum
  {\bf 3}, 010318}~(2022).

\bibitem{lay22}
David Layden.
\newblock ``First-order trotter error from a second-order perspective''.
\newblock \href{https://dx.doi.org/10.1103/PhysRevLett.128.210501}{Phys. Rev.
  Lett. {\bf 128}, 210501}~(2022).

\bibitem{som19}
Rolando~D Somma.
\newblock ``Quantum eigenvalue estimation via time series analysis''.
\newblock \href{https://dx.doi.org/10.1088/1367-2630/ab5c60}{New Journal of
  Physics {\bf 21}, 123025}~(2019).

\bibitem{cbk21}
Laura Clinton, Johannes Bausch, Joel Klassen, and Toby Cubitt.
\newblock ``Phase estimation of local hamiltonians on nisq hardware''~(2021).
\newblock  \href{http://arxiv.org/abs/2110.13584}{arXiv:2110.13584}.

\bibitem{ral21}
Patrick Rall.
\newblock ``Faster coherent quantum algorithms for phase, energy, and amplitude
  estimation''.
\newblock \href{https://dx.doi.org/10.22331/q-2021-10-19-566}{Quantum {\bf 5},
  566}~(2021).

\bibitem{bcc15}
Dominic~W Berry, Andrew~M Childs, Richard Cleve, Robin Kothari, and Rolando~D
  Somma.
\newblock ``Simulating hamiltonian dynamics with a truncated taylor series''.
\newblock Physical review letters {\bf 114}, 090502~(2015).
\newblock
  url:~\href{https://doi.org/10.1103/PhysRevLett.114.090502}{doi.org/10.1103/PhysRevLett.114.090502}.

\bibitem{lc17}
Guang~Hao Low and Isaac~L Chuang.
\newblock ``Optimal hamiltonian simulation by quantum signal processing''.
\newblock \href{https://dx.doi.org/10.1103/PhysRevLett.118.010501}{Physical
  review letters {\bf 118}, 010501}~(2017).

\bibitem{cmn18}
Andrew~M Childs, Dmitri Maslov, Yunseong Nam, Neil~J Ross, and Yuan Su.
\newblock ``Toward the first quantum simulation with quantum speedup''.
\newblock \href{https://dx.doi.org/10.1073/pnas.1801723115}{Proceedings of the
  National Academy of Sciences {\bf 115}, 9456--9461}~(2018).

\bibitem{lc19}
Guang~Hao Low and Isaac~L Chuang.
\newblock ``Hamiltonian simulation by qubitization''.
\newblock \href{https://dx.doi.org/10.22331/q-2019-07-12-163}{Quantum {\bf 3},
  163}~(2019).

\bibitem{kos07}
Emanuel Knill, Gerardo Ortiz, and Rolando~D Somma.
\newblock ``Optimal quantum measurements of expectation values of
  observables''.
\newblock \href{https://dx.doi.org/10.1103/PhysRevA.75.012328}{Physical Review
  A {\bf 75}, 012328}~(2007).

\bibitem{wbg20}
James~D. Watson, Johannes Bausch, and Sevag Gharibian.
\newblock ``The complexity of translationally invariant problems beyond ground
  state energies''~(2020).
\newblock  \href{http://arxiv.org/abs/2012.12717}{arXiv:2012.12717}.

\bibitem{pmsy14}
Alberto Peruzzo, Jarrod McClean, Peter Shadbolt, Man-Hong Yung, Xiao-Qi Zhou,
  Peter~J Love, Al{\'a}n Aspuru-Guzik, and Jeremy~L O’brien.
\newblock ``A variational eigenvalue solver on a photonic quantum processor''.
\newblock \href{https://dx.doi.org/10.1038/ncomms5213}{Nature communications
  {\bf 5}, 1--7}~(2014).

\bibitem{mrba16}
Jarrod~R McClean, Jonathan Romero, Ryan Babbush, and Al{\'a}n Aspuru-Guzik.
\newblock ``The theory of variational hybrid quantum-classical algorithms''.
\newblock \href{https://dx.doi.org/10.1088/1367-2630/18/2/023023}{New Journal
  of Physics {\bf 18}, 023023}~(2016).

\bibitem{so12}
Attila Szabo and Neil~S Ostlund.
\newblock ``Modern quantum chemistry: introduction to advanced electronic
  structure theory''.
\newblock Courier Corporation. ~(2012).

\bibitem{gg21}
Sevag Gharibian and Fran\c{c}ois Le~Gall.
\newblock ``Dequantizing the quantum singular value transformation: Hardness
  and applications to quantum chemistry and the quantum pcp conjecture''.
\newblock In Proceedings of the 54th Annual ACM SIGACT Symposium on Theory of
  Computing.
\newblock \href{https://dx.doi.org/10.1145/3519935.3519991}{Pages 19--32}.
\newblock ~(2022).

\bibitem{cgj18}
Shantanav Chakraborty, Andr{\'a}s Gily{\'e}n, and Stacey Jeffery.
\newblock ``{The Power of Block-Encoded Matrix Powers: Improved Regression
  Techniques via Faster Hamiltonian Simulation}''.
\newblock In Christel Baier, Ioannis Chatzigiannakis, Paola Flocchini, and
  Stefano Leonardi, editors, 46th International Colloquium on Automata,
  Languages, and Programming (ICALP 2019).
\newblock \href{https://dx.doi.org/10.4230/LIPIcs.ICALP.2019.33}{Volume 132 of
  Leibniz International Proceedings in Informatics (LIPIcs), pages
  33:1--33:14}.
\newblock Dagstuhl, Germany~(2019). Schloss Dagstuhl--Leibniz-Zentrum fuer
  Informatik.

\bibitem{gslw19}
Andr{\'a}s Gily{\'e}n, Yuan Su, Guang~Hao Low, and Nathan Wiebe.
\newblock ``Quantum singular value transformation and beyond: exponential
  improvements for quantum matrix arithmetics''.
\newblock In Proceedings of the 51st Annual ACM SIGACT Symposium on Theory of
  Computing.
\newblock \href{https://dx.doi.org/10.1145/3313276.3316366}{Pages 193--204}.
\newblock ~(2019).

\bibitem{ral20}
Patrick Rall.
\newblock ``Quantum algorithms for estimating physical quantities using block
  encodings''.
\newblock \href{https://dx.doi.org/10.1103/PhysRevA.102.022408}{Physical Review
  A {\bf 102}, 022408}~(2020).

\bibitem{tong2021fast}
Yu~Tong, Dong An, Nathan Wiebe, and Lin Lin.
\newblock ``Fast inversion, preconditioned quantum linear system solvers, fast
  green's-function computation, and fast evaluation of matrix functions''.
\newblock \href{https://dx.doi.org/10.1103/PhysRevA.104.032422}{Physical Review
  A {\bf 104}, 032422}~(2021).

\bibitem{rice2021quantum}
Julia~E Rice, Tanvi~P Gujarati, Mario Motta, Tyler~Y Takeshita, Eunseok Lee,
  Joseph~A Latone, and Jeannette~M Garcia.
\newblock ``Quantum computation of dominant products in lithium--sulfur
  batteries''.
\newblock \href{https://dx.doi.org/10.1063/5.0044068}{The Journal of Chemical
  Physics {\bf 154}, 134115}~(2021).

\bibitem{helgaker2014molecular}
Trygve Helgaker, Poul Jorgensen, and Jeppe Olsen.
\newblock ``Molecular electronic-structure theory''.
\newblock \href{https://dx.doi.org/10.1002/9781119019572}{John Wiley \& Sons}.
  ~(2014).

\bibitem{seeley2012bravyi}
Jacob~T Seeley, Martin~J Richard, and Peter~J Love.
\newblock ``The bravyi-kitaev transformation for quantum computation of
  electronic structure''.
\newblock \href{https://dx.doi.org/10.1063/1.4768229}{The Journal of chemical
  physics {\bf 137}, 224109}~(2012).

\bibitem{hhl09}
Aram~W Harrow, Avinatan Hassidim, and Seth Lloyd.
\newblock ``Quantum algorithm for linear systems of equations''.
\newblock \href{https://dx.doi.org/10.1103/PhysRevLett.103.150502}{Physical
  review letters {\bf 103}, 150502}~(2009).

\bibitem{cks17}
Andrew~M Childs, Robin Kothari, and Rolando~D Somma.
\newblock ``Quantum algorithm for systems of linear equations with
  exponentially improved dependence on precision''.
\newblock \href{https://dx.doi.org/10.1137/16M1087072}{SIAM Journal on
  Computing {\bf 46}, 1920--1950}~(2017).

\bibitem{blrm19}
Carlos Bravo-Prieto, Ryan LaRose, M.~Cerezo, Yigit Subasi, Lukasz Cincio, and
  Patrick~J. Coles.
\newblock ``Variational quantum linear solver''~(2019).
\newblock  \href{http://arxiv.org/abs/1909.05820}{arXiv:1909.05820}.

\bibitem{hbr19}
Hsin-Yuan Huang, Kishor Bharti, and Patrick Rebentrost.
\newblock ``Near-term quantum algorithms for linear systems of equations with
  regression loss functions''.
\newblock \href{https://dx.doi.org/10.1088/1367-2630/ac325f}{New Journal of
  Physics {\bf 23}, 113021}~(2021).

\bibitem{syso19}
Yi\u{g}it Suba\c{s}\i, Rolando~D Somma, and Davide Orsucci.
\newblock ``Quantum algorithms for systems of linear equations inspired by
  adiabatic quantum computing''.
\newblock \href{https://dx.doi.org/10.1103/PhysRevLett.122.060504}{Physical
  review letters {\bf 122}, 060504}~(2019).

\bibitem{an2019quantum}
Dong An and Lin Lin.
\newblock ``Quantum linear system solver based on time-optimal adiabatic
  quantum computing and quantum approximate optimization algorithm''.
\newblock \href{https://dx.doi.org/10.1145/3498331}{ACM Transactions on Quantum
  Computing{\bf 3}}~(2022).

\bibitem{lin2020optimal}
Lin Lin and Yu~Tong.
\newblock ``Optimal polynomial based quantum eigenstate filtering with
  application to solving quantum linear systems''.
\newblock \href{https://dx.doi.org/10.22331/q-2020-11-11-361}{Quantum {\bf 4},
  361}~(2020).

\bibitem{sb13}
Rolando~D Somma and Sergio Boixo.
\newblock ``Spectral gap amplification''.
\newblock \href{https://dx.doi.org/10.1137/120871997}{SIAM Journal on Computing
  {\bf 42}, 593--610}~(2013).

\bibitem{aa17}
Yosi Atia and Dorit Aharonov.
\newblock ``Fast-forwarding of hamiltonians and exponentially precise
  measurements''.
\newblock \href{https://dx.doi.org/10.1038/s41467-017-01637-7}{Nature
  communications {\bf 8}, 1--9}~(2017).

\bibitem{bfs11}
Brielin Brown, Steven~T Flammia, and Norbert Schuch.
\newblock ``Computational difficulty of computing the density of states''.
\newblock \href{https://dx.doi.org/10.1103/PhysRevLett.107.040501}{Physical
  review letters {\bf 107}, 040501}~(2011).

\bibitem{jgl10}
Stephen~P Jordan, David Gosset, and Peter~J Love.
\newblock ``Quantum-merlin-arthur--complete problems for stoquastic
  hamiltonians and markov matrices''.
\newblock \href{https://dx.doi.org/10.1103/PhysRevA.81.032331}{Physical Review
  A {\bf 81}, 032331}~(2010).

\bibitem{gs15}
Sevag Gharibian and Jamie Sikora.
\newblock ``Ground state connectivity of local hamiltonians''.
\newblock \href{https://dx.doi.org/10.1145/3186587}{ACM Trans. Comput.
  Theory{\bf 10}}~(2018).

\bibitem{watson2021complexity}
James~D. Watson and Johannes Bausch.
\newblock ``The complexity of approximating critical points of quantum phase
  transitions''~(2021).
\newblock  \href{http://arxiv.org/abs/2105.13350}{arXiv:2105.13350}.

\end{thebibliography}

\appendix
\section{Ground State Energy Estimation}\label{sec:lt21_details}
In this section, we review the techniques in 
\cite{lt21}, 
which proposed a hybrid quantum/classical algorithm for estimating the ground state energy of a Hamiltonian. Compared with the algorithms in previous works, the algorithm in \cite{lt21} uses fewer quantum resources and does not need to access the block-encoding of the Hamiltonian.

First of all, they assumed that  the given initial state $\ket{\phi_0}$\footnote{In \cite{lt21}, they allowed the initial state to be a mixed state. For simplicity, we still denote it as $\ket{\phi_0}$.} has a nontrivial overlap with the ground state of $H$.

\subsection{Quantum part of the algorithm}\label{sec:quantum_hadamard}

Fix $j\in \mathbb{Z}$. Suppose we want to estimate $\Re(\bra{\phi_0}e^{-ij\tau H}\ket{\phi_0})$. Then, we set $\mathrm{W}=I$ and define a random variable $X_j$ as follows:
\begin{align*}
    X_j:=\begin{cases}1 & \text{if the outcome is}~0\\
    -1 & \text{if the outcome is}~1
    \end{cases}.
\end{align*}
Since the state before the measurement is
\begin{align}
    \frac{1}{2} (\ket{0}\otimes (I + e^{-ij\tau H})\ket{\phi_0} + \ket{1}\otimes (I - e^{-ij\tau H})\ket{\phi_0}), 
\end{align}
we have
\begin{align}\label{eq:expect_real_part}
    \E[X_j]=&~ \Pr[X_j=0]-\Pr[X_j=1]\notag\\
    = &~ \frac{1}{4}\bra{\phi_0} (I + e^{ij\tau H})(I + e^{-ij\tau H}) \ket{\phi_0} - \frac{1}{4}\bra{\phi_0} (I - e^{ij\tau H})(I - e^{-ij\tau H}) \ket{\phi_0}\notag\\
    = &~ \frac{1}{2}\bra{\phi_0} (e^{ij\tau H} + e^{-ij\tau H})\ket{\phi_0}\notag\\
    = &~ \Re (\bra{\phi_0}e^{-ij\tau H}\ket{\phi_0}).
\end{align}
For the imaginary part $\Im(\bra{\phi_0}e^{-ij\tau H}\ket{\phi_0})$, we can set $W$ to be the phase gate $\begin{bmatrix} 1 & 0\\0 & -i \end{bmatrix}$ and define the random variable $Y_j$ similarly. Then, we have
\begin{align}\label{eq:expect_img_part}
    \E[Y_j] = \Im (\bra{\phi_0}e^{-ij\tau H}\ket{\phi_0}).
\end{align}

Therefore, Eqs.~\eqref{eq:expect_real_part} and \eqref{eq:expect_img_part} implies the following claim:
\begin{claim}[Estimator of the Hamiltonian expectation]\label{clm:estimator_expectation}
For any $j\in \Z$, the random variable $X_j + i Y_j$ is an un-biased estimator for $\bra{\phi_0}e^{-ij\tau H}\ket{\phi_0}$.
\end{claim}

\subsection{Classical part of the algorithm}
Let $\tau$ be a normalization factor such that $\|\tau H\|\leq \pi/3$.
Suppose the initial state $\ket{\phi_0}$ can be decomposed in the eigenspace of $H$ as $\ket{\phi_0} = \sum_{k} \sqrt{p_k} \ket{\psi_k}$. Let $p(x)$ be the following density function (spectral measure):
\begin{align}
    p(x) := \sum_k p_k \delta(x - \tau \lambda_k)~~~\forall x\in [-\pi, \pi].
\end{align}
That is, $p(x)$ is the distribution of the state energy with respect to $\tau H$ after we measure $\ket{\phi_0}$ in the eigenbasis of $H$.

Define the $2\pi$-periodic Heaviside function by
\begin{align}\label{eq:def_heaviside}
    H(x) = \begin{cases}
        1 & x\in [2k\pi, (2k+1)\pi)\\
        0 & x \in [(2k-1)\pi, 2k\pi)
    \end{cases}~~~\forall k\in \mathbb{Z}.
\end{align}
Then, we define the $2\pi$-periodic CDF of $p$ as the convolution of $H$ and $p$:
\begin{align}
    C(x) := (H * p)(x).
\end{align}
For any $x\in [-\pi/3, \pi/3]$, for any $w\in \mathbb{Z}$, we have
\begin{align}
    C(x + 2w\pi) = &~ \int_{-\pi}^\pi H(x+2w\pi - t) p(t) \d t\\
    = &~ \sum_{k} p_k\cdot \int_{-\pi}^\pi H(x+2w\pi - t) \delta(t - \tau \lambda_k)\d t\notag\\
    = &~ \sum_{k} p_k\cdot H(x + 2w \pi - \tau \lambda_k)\notag\\
    = &~ \sum_{k} p_k \cdot \mathbf{1}_{x \geq \tau \lambda_k}\notag\\
    = &~ \sum_{k: \lambda_k \leq x} p_k,
\end{align}
where the first step follows from the definition of convolution, the second step follows from Dirac delta function's property, and the third step follows from $H$ has period $2\pi$. We note that $C(x)$ is right continuous and non-decreasing in $[-\pi/3, \pi/3]$. 

However, we cannot directly evaluate $C(x)$, but we can approximate it! Define the approximate CDF (ACDF) as
\begin{align}\label{eq:def_acdf}
    \wt{C}(x) := (F * p) (x),
\end{align}
where $F(x) = \sum_{|j|\leq d} \hat{F}_j e^{ijx}$ is a low Fourier-degree approximation of the Heaviside function $H(x)$ such that
\begin{align}
    |F(x) - H(x)|\leq \epsilon~~~\forall x\in [-\pi+\delta, -\delta]\cup [\delta, \pi - \delta].
\end{align}
The construction of $F$ is given by Lemma~\ref{lem:approx_Heaviside}. 
Furthermore, the approximation error of $\wt{C}(x)$ is bounded by
\begin{align}
    C(x-\delta) - \epsilon \leq \wt{C}(x) \leq C(x + \delta) + \epsilon,
\end{align}
for any $x\in [-\pi/3, \pi/3]$, $\delta \in (0, \pi/6)$ and $\epsilon>0$. 

\subsubsection{Estimating the ACDF}
The goal of this section is to prove Lemma~\ref{lem:est_acdf}, which constructs an estimator for $\wt{C}(x)$ (defined by Eq.~\eqref{eq:def_acdf}).
\begin{lemma}[Estimating the ACDF]\label{lem:est_acdf}
For any $\sigma>0$, for any $x\in [-\pi, \pi]$, there exists an un-biased estimator $\ov{G}(x)$ for the ACDF $\wt{C}(x)$ with variance at most $\sigma^2$. 

Furthermore, $\ov{G}(x)$ runs the quantum circuit (Figure~\ref{fig:hadamard_test}) $O(\frac{\log^2 d}{\sigma^2})$ times with expected total evolution time $O(\frac{\tau d\log d}{\sigma^2})$.
\end{lemma}

\begin{proof}
$\wt{C}(x)$ can be expanded in the following way:
\begin{align}
    \wt{C}(x) = &~ (F * p)(x)\\
    = &~ \int_{-\pi}^\pi F(x-y) p(y) \d y\notag\\
    = &~ \sum_{|j|\leq d} \int_{-\pi}^\pi \hat{F}_j e^{ij(x-y)} p(y)\d y\notag\\
    = &~ \sum_{|j|\leq d} \hat{F}_j e^{ijx} \int_{-\pi}^\pi p(y) e^{-ijy}\d y\notag\\
    = &~ \sum_{|j|\leq d} \hat{F}_j e^{ijx} \sum_k p_k e^{-ij\tau \lambda_k}\notag\\
    = &~ \sum_{|j|\leq d} \hat{F}_j e^{ijx} \cdot \bra{\phi_0} e^{-ij\tau H} \ket{\phi_0},
\end{align}
where the third step follows from the Fourier expansion of $F(x-y)$, the fifth step follows from the property of Dirac's delta function, and the last step follows from the definition of $p_k$ and the eigenvalues of matrix exponential.

To estimate $\bra{\phi_0} e^{-ij\tau H} \ket{\phi_0}$, we use the multi-level Monte Carlo method. Define a random variable $J$ with support $\{-d, \cdots, d\}$ such that
\begin{align}\label{eq:def_J}
    \Pr[J=j]=\left|\hat{F}_j\right|/{\cal F},
\end{align}
where ${\cal F}:=\sum_{|j|\leq d}|\hat{F}_j|$. Then, let $Z:=X_J + i Y_J\in \{\pm 1\pm i\}$. Define an estimator $G(x; J, Z)$ as follows:
\begin{align*}
    G(x; J, Z):={\cal F}\cdot Z e^{i(\theta_J + Jx)},
\end{align*}
where $\theta_j$ is defined by $\hat{F}_j = |\hat{F}_j|e^{i\theta_j}$. Then, we show that $G(x; J, Z)$ is un-biased:
\begin{align*}
    \E[G(x; J, Z)] = &~ \sum_{|j|\leq d} \E\left[(X_j + iY_j)e^{i(\theta_j + jx)}|\hat{F}_j|\right]\\
    = &~ \sum_{|j|\leq d} \hat{F}_j e^{ijx} \cdot \E\left[X_j + iY_j\right]\\
    = &~ \sum_{|j|\leq d} \hat{F}_j e^{ijx} \cdot \bra{\phi_0} e^{-ij\tau H} \ket{\phi_0}\\
    = &~ \wt{C}(x),
\end{align*}
where the third step follows from Claim~\ref{clm:estimator_expectation}. Moreover, the variance of $G$ can be upper-bounded by:
\begin{align*}
    \mathrm{Var}[G(x; J, Z)]= &~ \E[|G(x; J, Z)|^2] - |\E[G(x; J, Z)]|^2\\
    \leq &~ \E[|G(x; J, Z)|^2]\\
    = &~ {\cal F}^2 \cdot \E[|X_J + i Y_J|^2]\\
    = &~ 2{\cal F}^2,
\end{align*}
where the third step follows from $|e^{i(\theta_J+Jx)}|=1$, and the last step follows from $X_j, Y_j\in \{\pm 1\}$.

Hence, we can take $N_s:=\frac{2{\cal F}^2}{\sigma^2}$ independent samples of $(J, Z)$, denoted by $\{(J_k, Z_k)\}_{k\in [N_s]}$ and compute
\begin{align*}
    \ov{G}(x) := \frac{1}{N_s}\sum_{k=1}^{N_s} G(x; J_k, Z_k).
\end{align*}
Then, we have
\begin{align*}
    \E[\ov{G}(x)] = \wt{C}(x),~~\text{and}~~\mathrm{Var}[\ov{G}(x)]\leq \sigma^2.
\end{align*}

The expected total evolution time is
\begin{align*}
    {\cal T}_{\mathsf{tot}} := N_s \tau \E[|J|]= \frac{2{\cal F}^2}{\sigma^2} \tau \sum_{|j|\leq d}|j|\cdot \frac{|\hat{F}_j| }{{\cal F}}= \frac{2{\cal F}\tau}{\sigma^2}\sum_{|j|\leq d}|j||\hat{F}_j|.
\end{align*}

By Lemma~\ref{lem:approx_Heaviside}, we know that $|\hat{F}_j|=O(1/|j|)$. Hence,
we have ${\cal F} = \sum_{|j|\leq d}O(1/|j|) = O(\log d)$. Thus, the number of samples is
\begin{align*}
    N_s = O\left(\frac{\log^2 d}{\sigma^2}\right).
\end{align*}
And the expected total evolution time is
\begin{align*}
    {\cal T}_{\mathsf{tot}} = O\left(\frac{\tau d\log d}{\sigma^2}\right).
\end{align*}
The lemma is then proved.
\end{proof}

\subsubsection{Inverting the CDF}
We first define the CDF inversion problem:
\begin{definition}[The CDF inversion problem]\label{def:inv_cdf}
For $0 < \delta < \pi/6$, $0 < \eta < 1$, find $x^\star\in (-\pi/3, \pi/3)$ such that
\begin{align*}
    C(x^\star + \delta) > \eta /2, \quad C(x^\star-\delta) < \eta.
\end{align*}
\end{definition}
\begin{remark}
The condition in Definition~\ref{def:inv_cdf} is weaker than $\eta/2<C(x)<\eta$ due to the discontinuity of $C(x)$. For any CDF $C(x)$, such an $x^\star$ must exist: let $a := \sup~\{x\in (-\pi/3, \pi/3): C(x) \leq \eta/2\}$ and $b:=\inf~\{x\in (-\pi/3, \pi/3): C(x) \geq \eta\}$. Since $C(x)$ is non-decreasing, we have $a\leq b$. And any $x\in (a-\delta, b+\delta)$ satisfies the condition in Definition~\ref{def:inv_cdf}.
\end{remark}

Then, we give an algorithm that solves the CDF inversion problem.
\begin{lemma}[Inverting the CDF, Theorem 2 in \cite{lt21}]\label{lem:invert_cdf}
There exists an algorithm that solves the CDF inversion problem (Definition~\ref{def:inv_cdf}) with probability at least $1-\nu$ such that:
\begin{enumerate}
    \item the number of independent samples of $(J,Z)$ is
    \begin{align*}
        O\left(\eta^{-2}\cdot (\log(\nu^{-1})+\log\log(\delta^{-1}))\cdot (\log(\delta^{-1})+\log \log (\delta^{-1}\eta^{-1}))^2\right)
    \end{align*}
    \item the expected total evolution time is
    \begin{align*}
        O\left(\tau \eta^{-2}\cdot \delta^{-1}\log(\delta^{-1}\eta^{-1})\cdot  (\log(\delta^{-1})+\log\log(\delta^{-1}\eta^{-1}))\cdot (\log(\nu^{-1})+\log\log(\delta^{-1})) \right)
    \end{align*}
    \item the maximal evolution time is
    \begin{align*}
        O\left(\tau \delta^{-1}\log(\delta^{-1}\eta^{-1})\right)
    \end{align*}
    \item the classical running time is
    \begin{align*}
       O\left(\eta^{-2}\log(\delta^{-1})\cdot (\log(\nu^{-1})+\log\log(\delta^{-1}))\cdot (\log(\delta^{-1})+\log \log (\delta^{-1}\eta^{-1}))^2\right).
    \end{align*}
\end{enumerate}
\end{lemma}

\begin{proof}
For any $x\in [-\pi/3, \pi/3]$, at least one of the following conditions will hold:
\begin{align}\label{eq:cond_dec_cdf_inv}
    C(x + \delta) > \eta /2, ~~\text{or}~~C(x-\delta) < \eta.
\end{align}
Suppose we have a sub-routine $\textsc{Certify}(x, \delta, \eta, \{J_k, Z_k\})$ such that if $C(x + \delta) > \eta /2$, it returns 0; otherwise, it returns 1.

Then, we can solve the CDF inversion problem via the binary search (Algorithm~\ref{alg:cdf_inv}).

\begin{algorithm}[t]
\caption{Inverting the CDF}
\label{alg:cdf_inv} 
\begin{algorithmic}[1]
    \algrenewcommand\algorithmicprocedure{\textbf{procedure}}
	\Procedure{InvertCDF}{$\eta, \delta, \{J_k, Z_k\}$}
        \State $x_{L}\gets -\pi/3$, $X_{R}\gets \pi/3$
        \While{$x_R - x_L > 2\delta$}\label{ln:binary_while}
            \State $x_M\gets (x_L + x_R)/2$
            \State $u\gets \textsc{Certify}(x_M, (2/3)\delta, \eta, \{J_k, Z_k\})$
            \If{$u=0$}
                \State $x_R\gets x_M + (2/3)\delta$
            \Else 
                \State $x_L\gets x_M - (2/3)\delta$
            \EndIf
        \EndWhile
        \State \Return $(x_L + x_R)/2$
	\EndProcedure
\end{algorithmic}
\end{algorithm}

In Line~\ref{ln:binary_while}, $x_L$ and $x_R$ always satisfy the following conditions:
\begin{align*}
    C(x_L)<\eta, \quad C(x_R)>\eta/2,
\end{align*}
which is guaranteed by $\textsc{Certify}(x_M, (2/3)\delta, \eta, \{J_k, Z_k\})$. Then, when the while-loop ends, we have $x_R-x_L\leq 2\delta$. Let $x^\star := (x_L + x_R)/2$ be the output of Algorithm~\ref{alg:cdf_inv}. Then, we get that
\begin{align*}
    C(x^\star +\delta) \geq &~ C(x_R) > \eta/2,\\
    C(x^\star -\delta) \leq &~ C(x_L) < \eta.
\end{align*}

And it is easy to see that Algorithm~\ref{alg:cdf_inv} will call $\textsc{Certify}$ $L:=O(\log(1/\delta))$ times. Then, by Lemma~\ref{lem:certify_two_cases} and union bound, Algorithm~\ref{alg:cdf_inv} will be correct with probability at least $1-\nu$. We note that different runs of \textsc{Certify} can share a same set of samples $\{J_k,Z_k\}$, which does not affect the union bound. Hence, the number of samples and the total evolution time follows directly from  Lemma~\ref{lem:certify_two_cases} and $d=O(\delta^{-1}\log(\delta^{-1}\eta^{-1}))$. 
\end{proof}

\begin{lemma}[\textsc{Certify} sub-routine]\label{lem:certify_two_cases}
For any $\nu>0$, there exists an algorithm that distinguishes the two cases in Eq.~\eqref{eq:cond_dec_cdf_inv} for any $x\in [-\pi/3, \pi/3]$ with probability at least $1-O(\nu/L)$ using 
\begin{align*}
    O\left(\eta^{-2}\log^2(d)(\log(1/\nu)+\log\log(1/\delta))\right)
\end{align*}
independent samples of $(J,Z)$, and total evolution time
\begin{align*}
    O\left(\eta^{-2}\tau d \log(d)(\log(1/\nu)+\log\log(1/\delta)) \right)
\end{align*}
in expectation.
\end{lemma}

\begin{proof}
To decide which one of the conditions holds for $x$, we can estimate the ACDF $\wt{C}(x)$. If we take $\epsilon=\eta/8$ in Lemma~\ref{lem:approx_Heaviside}, then the constructed ACDF satisfies
\begin{align*}
    C(x-\delta) - \eta/8 \leq \wt{C}(x)\leq C(x+\delta) + \eta/8.
\end{align*}
Thus,
\begin{align*}
    &\wt{C}(x) > (5/8)\eta ~~\Rightarrow~~C(x+\delta)>\eta/2,\\
    &\wt{C}(x) < (7/8)\eta ~~\Rightarrow~~C(x-\delta)<\eta.
\end{align*}
Then, we can distinguish $\wt{C}(x) > (5/8)\eta$ or $\wt{C}(x) < (7/8)\eta$ by the estimator in Lemma~\ref{lem:est_acdf}.

\begin{algorithm}[t]
\caption{Distinguish the two cases in Eq.~\eqref{eq:cond_dec_cdf_inv}}
\label{alg:cdf_dis} 
\begin{algorithmic}[1]
    \algrenewcommand\algorithmicprocedure{\textbf{procedure}}
	\Procedure{Certify}{$x,\eta, \delta, \{J_k, Z_k\}$}
        \State $c\gets 0$, $N_b\gets \Omega(\log(1/\nu)+\log\log(1/\delta))$
        \For{$1\leq r\leq N_b$}
            \State Compute $\ov{G}(x)$ using $\{J_k, Z_k\}_{k\in [(r-1)N_s+1, rN_s]}$\Comment{Lemma~\ref{lem:est_acdf}}
            \If{$\ov{G}(x)\geq (3/4)\eta$}
                \State $c\gets c+1$
            \EndIf
        \EndFor
        \State \Return $\mathbf{1}_{c\leq N_b/2}$
	\EndProcedure
\end{algorithmic}
\end{algorithm}

In Algorithm~\ref{alg:cdf_dis}, we compute the estimator $\ov{G}(x)$ $N_b$ times independently, where each time we use $N_s$ samples of $(J, Z)$. We note that an error occurs when $\wt{C}(x)>(7/8)\eta$ but $\ov{G}(x) < (3/4)\eta$, or $\wt{C}(x)<(5/8)\eta$ but $\ov{G}(x) > (3/4)\eta$ (when $(5/8)\eta \leq \wt{C}(x)\leq (7/8)\eta$, any output is correct). By Chebyshev's inequality, we have
\begin{align*}
    \Pr[\ov{G}(x)~\text{has an error}]\leq &~\Pr\left[\ov{G}(x) < \frac{3}{4}\eta ~\Big|~ \wt{C}(x)>\frac{7}{8}\eta\right] + \Pr\left[\ov{G}(x) > \frac{3}{4}\eta ~\Big|~ \wt{C}(x)<\frac{5}{8}\eta\right]\\
    \leq &~ 2\cdot \frac{\sigma^2}{\eta^2/64}\\
    \leq &~ \frac{1}{4},
\end{align*}
if we take $\sigma^2=O(\eta^2)$ in Lemma~\ref{lem:est_acdf}.

Then, by the Chernoff bound, we have 
\begin{align*}
    \Pr[\textsc{Certify}~\text{makes an error}]\leq \exp(-\Omega(N_b)) \leq \nu/L,
\end{align*}
if we take $N_b := \Omega(\log (L/\nu))=\Omega(\log(1/\nu)+\log\log(1/\delta))$.
Thus, the total number of samples is
\begin{align*}
    N_b N_s = O\left(\eta^{-2}\log^2(d)(\log(1/\nu)+\log\log(1/\delta))\right),
\end{align*}
and the expected total evolution time is
\begin{align*}
O\left(\eta^{-2}\tau d \log(d)(\log(1/\nu)+\log\log(1/\delta)) \right),
\end{align*}
which complete the proof of the lemma.
\end{proof}
\subsubsection{Estimating the ground state energy}

\begin{corollary}[Ground state energy estimation, Corollary 3 in \cite{lt21}]
If $p_0\geq \eta$ for some known $\eta$, then with probability at least $1-\nu$, the ground state energy $\lambda_0$ can be estimated within additive error $\epsilon$, such that:
\begin{enumerate}
    \item the number of times running the quantum circuit (Figure~\ref{fig:hadamard_test}) is $\wt{O}(\eta^{-2})$.
    \item the expected total evolution time is $\wt{O}(\epsilon^{-1}\eta^{-2})$. 
    \item the maximal evolution time is $\wt{O}(\epsilon^{-1})$.
    \item the classical running time is $\wt{O}(\eta^{-2})$.
\end{enumerate}
\end{corollary}
\begin{proof}
Suppose we can solve the CDF inversion problem (Definition~\ref{def:inv_cdf}) for $\delta = \tau \epsilon$ and $\eta$, i.e., we find an $x^\star$ such that
\begin{align*}
    C(x^\star + \tau \epsilon) > \eta/2 > 0, ~~~C(x^\star - \tau \epsilon) < \eta \leq p_0.
\end{align*}
Since $C(x)$ cannot take value between $0$ and $p_0$, we have
\begin{align*}
    x^\star + \tau \epsilon \geq \tau \lambda_0, ~~~ x^\star - \tau \epsilon < \tau \lambda_0,
\end{align*}
which is 
\begin{align*}
    |x^\star/\tau  - \lambda_0|\leq \epsilon.
\end{align*}

The costs of this algorithm follows from Lemma~\ref{lem:invert_cdf}.
\end{proof}

\subsection{Low Fourier degree approximation of the Heaviside function}
We construct the low degree approximation of the Heaviside function in this section.\footnote{The construction in \cite{lt21} is not enough to prove Lemma~\ref{lem:approx_acdf} because the range of $F_{d,\delta}$ is $[-\epsilon/2, 1+\epsilon]$ while Lemma~\ref{lem:approx_acdf} requires the range to be $[0,1]$. We fix this issue in Lemma~\ref{lem:approx_Heaviside}.}
\begin{lemma}[Constructing low degree approximation of $H$]\label{lem:approx_Heaviside}
Let $H(x)$ be the $2\pi$-period Heaviside function (Eq.~\eqref{eq:def_heaviside}). For any $\delta \in (0, \pi/2)$ such that $\tan(\delta/2)\leq 1-1/\sqrt{2}$, there exists a $d=O(\delta^{-1}\log(\delta^{-1}\epsilon^{-1}))$ and a $2\pi$-period function $F_{d,\delta}(x)$ of the form:
\begin{align}
    F_{d,\delta}(x) = \frac{1}{\sqrt{2\pi}}\sum_{j=-d}^d \wh{F}_{d,\delta,j} \cdot e^{ijx}
\end{align}
such that
\begin{enumerate}
    \item $F_{d,\delta}(x)\in [0, 1]$ for all $x\in \R$.
    \item $|F_{d,\delta}(x) - H(x)|\leq \epsilon$ for $x\in [-\pi + \delta, -\delta]\cup [\delta, \pi - \delta]$.
    \item $|\wh{F}_{d,\delta,j}|=\Theta(1/|j|)$ for $j\ne 0$.
\end{enumerate}
\end{lemma}

\begin{proof}
We first construct $F_{d,\delta}'(x)$ by mollifying the Heaviside function with $M_{d,\delta}(x)$ in Lemma~\ref{lem:mollifier}:
\begin{align}
    F'_{d,\delta}(x) := (M_{d,\delta} * H)(x) = \int_{-\pi}^{\pi} M_{d,\delta}(y) H(x-y) \d y.
\end{align}

We can verify that $F'_{d,\delta}$ has Fourier degree at most $d$. It follows from the Chebyshev polynomial $T_d(x)$ is of degree $d$. Hence, the Fourier coefficients of $M_{d,\delta}(x)$:
\begin{align}
    \wh{M}_{d,\delta, j} = \frac{1}{\sqrt{2\pi}}\int_{-\pi}^\pi M_{d,\delta}(x)e^{-ijx}\d x \ne 0
\end{align}
only if $j\in \{-d,\dots, d\}$. Since $F_{d,\delta}$ is a convolution of $M_{d,\delta}$ and $H$, we have
\begin{align}
    \wh{F'}_{d,\delta,j}=\sqrt{2\pi} \wh{M}_{d,\delta, j} \wh{H}_{j}~~~\forall |j|\leq d.
\end{align}

Then, we define
\begin{align}
F_{d,\delta}(x) := \frac{1}{\sqrt{2\pi}}\sum_{j=-d}^d \wh{F}_{d,\delta,j} \cdot e^{ijx},    
\end{align}
where
\begin{align}
    \wh{F}_{d,\delta,j} = \begin{cases}
        \frac{1}{1+(5/4)\epsilon}\left(\wh{F'}_{d,\delta, j} + \sqrt{2\pi}\epsilon/4\right) & \text{if}~j=0,\\
        \frac{1}{1+(5/4)\epsilon}\wh{F'}_{d,\delta, j} & \text{otherwise}.
    \end{cases}
\end{align}
It is easy see that 
\begin{align}
    F_{d,\delta}(x) = \frac{F'_{d,\delta}(x) + \epsilon/4}{1+(5/4)\epsilon}~~~\forall x\in \R.
\end{align}

Then, we will show that taking $d=O(\delta^{-1}\log(\delta^{-1}\epsilon^{-1}))$ is enough to satisfy (1)-(3).

\paragraph{Part (1):}
We first compute the range of $F'_{d,\delta}(x)$:
\begin{align}
    F'_{d,\delta}(x)\leq \int_{-\pi}^{\pi} |M_{d,\delta}(y)|\d y\leq 1+\frac{4\pi}{{\cal N}_{d,\delta}},
\end{align}
where the second step follows from (2) in Lemma~\ref{lem:mollifier}.
On the other hand,
\begin{align*}
    F'_{d,\delta}(x)\geq -\frac{1}{{\cal N}_{d,\delta}} \int_{-\pi}^\pi H(y)\d y = \frac{-\pi}{{\cal N}_{d,\delta}}.
\end{align*}
Hence, if we take $d=O(\delta^{-1}\log(\delta^{-1}\epsilon^{-1}))$ such that
\begin{align}\label{eq:d_N}
     {\cal N}_{d,\delta}\geq C_1 e^{d\delta/\sqrt{2}}\sqrt{\frac{\delta}{d}}\cdot \mathrm{erf}(C_2\sqrt{d}\delta)\geq \frac{4\pi}{\epsilon} 
\end{align}
holds, we will have
\begin{align}\label{eq:F_upper_bound}
    -\epsilon/4 \leq F'_{d,\delta} \leq 1+\epsilon.
\end{align}
Therefore, for all $x\in \R$,
\begin{align}
    F_{d,\delta}(x) = \frac{F'_{d,\delta}(x) + \epsilon/4}{1+(5/4)\epsilon}\in [0,1].
\end{align}

\paragraph{Part (2):} The approximation error of $F'_{d,\delta}$ is 
\begin{align}
    |F'_{d,\delta}(x) - H(x)| \leq &~ \left|\int_{-\pi}^{\pi} M_{d,\delta}(y) (H(x-y)-H(x))\d y\right|\notag\\
    \leq &~ \int_{-\pi}^{\pi} |M_{d,\delta}(y)| |H(x-y)-H(x)|\d y,
\end{align}
where the first step follows from (2) in Lemma~\ref{lem:mollifier}, and the second step follows from the triangle inequality.

Fix $x\in [-\pi+\delta, -\delta]\cup [\delta, \pi-\delta]$. If $y\in (-\delta, \delta)$, then $H(x-y) = H(x)$ and 
\begin{align}
    \int_{-\delta}^{\delta} |M_{d,\delta}(y)| |H(x-y)-H(x)|\d y = 0.
\end{align}
If $|y|\geq \delta$, by (1) in Lemma~\ref{lem:mollifier}, we have $|M_{d,\delta}|\leq \frac{1}{{\cal N}_{d,\delta}}$. Since $|H(x-y)-H(x)|\leq 1$, we have
\begin{align}
    \left(\int_{-\pi}^{-\delta} + \int_{\delta}^{\pi}\right) |M_{d,\delta}(y)| |H(x-y)-H(x)|\d y \leq \frac{2\pi}{{\cal N}_{d,\delta}}\leq \epsilon/2,
\end{align}
where the last step follows from Eq.~\eqref{eq:d_N}. Therefore,
\begin{align}
    |F'_{d,\delta}(x) - H(x)|\leq \epsilon/2~~~\forall |x|\in [\delta, \pi-\delta].
\end{align}

Thus, 
\begin{align}
    |F_{d,\delta}(x) - H(x)| = &~ \left|\frac{F'_{d,\delta}(x) + \epsilon/4}{1+(5/4)\epsilon} - H(x)\right|\\
    \leq &~ |F'_{d,\delta}(x) - H(x)| + \frac{(5/4)\epsilon}{1+(5/4)\epsilon} |F'_{d,\delta}(x)| + \frac{\epsilon/4}{1+(5/4)\epsilon}\notag\\
    \leq &~ \epsilon/2 +  \frac{(5/4)\epsilon}{1+(5/4)\epsilon}(1+\epsilon) + \frac{\epsilon/4}{1+(5/4)\epsilon}\notag\\
    \leq &~ 2\epsilon,
\end{align}
where the second step follows from the triangle inequality, the third step follows from Eq.~\eqref{eq:F_upper_bound}. By scaling for $\epsilon$, we can make the approximation error at most $\epsilon$.

\paragraph{Part (3):}
Since $|\wh{F'}_{d,\delta,j}|=\sqrt{2\pi} |\wh{M}_{d,\delta, j}| |\wh{H}_{j}|$, we first bound $|\wh{M}_{d,\delta, j}|$:
\begin{align}
    \left|\wh{M}_{d,\delta, j}\right| \leq \frac{1}{\sqrt{2\pi}}\int_{-\pi}^\pi |M_{d,\delta}(x)|\d x\leq \frac{1}{\sqrt{2\pi}}\left(1+\frac{4\pi}{{\cal N}_{d,\delta}}\right)\leq \frac{1+\epsilon}{\sqrt{2\pi}},
\end{align}
where the second step follows from (2) in Lemma~\ref{lem:mollifier} and the last step follows from Eq.~\eqref{eq:d_N}.

For $|\wh{H}_j|$, if $j\ne 0$, we have
\begin{align}
    \wh{H}_j =  \frac{1}{\sqrt{2\pi}}\int_{-\pi}^{\pi} H(x) e^{-ijx}\d x
    =  \frac{1}{\sqrt{2\pi}}\int_{0}^{\pi} e^{-ijx}\d x
    =  \begin{cases}
    \frac{\sqrt{2}}{i\sqrt{\pi}j} & \text{if }j~\text{is odd},\\
    0 & \text{if }j~\text{is even}.
    \end{cases}
\end{align}
Hence, for $j\ne 0$,
\begin{align}
    |\wh{F'}_{d,\delta, j}|\leq \sqrt{2\pi} \cdot \frac{1+\epsilon}{\sqrt{2\pi}}\cdot \sqrt{\frac{2}{\pi}}\frac{1}{|j|}=\frac{1+\epsilon}{\sqrt{\pi/2}|j|}.
\end{align}
Then, by definition, we get that
\begin{align}
    |\wh{F}_{d,\delta, j}|\leq \frac{1+\epsilon}{\sqrt{\pi/2}(1+(5/4)\epsilon)|j|}= \Theta(1/|j|).
\end{align}

The proof of the lemma is completed.
\end{proof}

The following lemma shows the approximation ratio of the ACDF $\wt{C}(x)$ constructed from the low degree approximated Heaviside function $F(x)$ by Lemma~\ref{lem:approx_Heaviside}.

\begin{lemma}[Approximation ratio of the ACDF]\label{lem:approx_acdf}
For any $\epsilon>0$, $0<\delta < \pi/6$, let $F(x) := F_{d,\delta}(x)$ constructed by Lemma~\ref{lem:approx_Heaviside}. Then, for any $x\in [-\pi/3, \pi/3]$, the ACDF $\wt{C}(x) = (F*p)(x)$ satisfies:
\begin{align*}
    C(x-\delta) -\epsilon \leq \wt{C}(x) \leq C(x + \delta) + \epsilon.
\end{align*}
\end{lemma}
\begin{proof}
By (2) in Lemma~\ref{lem:approx_Heaviside}, we have
\begin{align}
    |F(x) - H(x)|\leq \epsilon ~~~\forall x\in [-\pi + \delta, -\delta]\cup [\delta, \pi - \delta].
\end{align}
Define $F_L := F(x - \delta)$ such that
\begin{align}\label{eq:F_L_1d}
    |F_L(x) - H(x)|\leq \epsilon ~~~\forall x\in [-\pi + 2\delta, 0]\cup [2\delta, \pi].
\end{align}
For $\wt{C}_L(x) := (F_L * p)(x)$, we have $\wt{C}_L(x) = \wt{C}(x - \delta)$, and for $x \in [-\pi/3, \pi/3]$,
\begin{align}
    |C(x) - \wt{C}_L(x)| =&~ \left|\int_{-\pi}^{\pi} p(x-y) (H(y) - F_L(y)) \d y\right|\\
    \leq &~ \int_{-\pi}^{\pi} p(x-y) |H(y) - F_L(y)| \d y\notag\\
    = &~ \left(\int_{-\pi}^0 + \int_{2\delta}^\pi\right) p(x-y)|H(y)-F_L(y)|\d y+ \int_0^{2\delta} p(x-y)|H(y)-F_L(y)|\d y\notag\\
    \leq &~ \epsilon\cdot \left(\int_{-\pi}^0 + \int_{2\delta}^\pi\right) p(x-y)\d y + \int_0^{2\delta} p(x-y)|H(y)-F_L(y)|\d y\notag\\
    \leq &~ \epsilon + \int_0^{2\delta} p(x-y)|H(y)-F_L(y)|\d y\notag\\
    \leq &~ \epsilon + \int_0^{2\delta} p(x-y) \d y\notag\\
    = &~ \epsilon + \int_{x-2\delta}^{x} p(y) \d y\notag\\
    = &~ \epsilon + C(x) - C(x-2\delta),
\end{align}
where the second step follows from Cauchy-Schwarz inequality, the forth step follows from Eq.~\eqref{eq:F_L_1d}, the fifth step follows from $p(x)$ is a density function, the sixth step follows from $H(y) = 1$ and $F_L(y)\in [0,1]$ for $y\in [0, 2\delta]$, the last step follows from $C(x)$ is the CDF of $p(x)$ in $[-\pi, \pi]$. 

Hence, we have
\begin{align}
    \wt{C}_L(x) \geq C(x) - (\epsilon + C(x) - C(x - 2\delta)) = C(x - 2\delta) - \epsilon,
\end{align}
which proves the first inequality: 
\begin{align}
    \wt{C}(x - \delta) \geq C(x-2\delta) - \epsilon.
\end{align}

Similarly, we can define $F_R := F(x + \delta)$ and $\wt{C}_R(x) := (F_R * p)(x)$. We can show that
\begin{align}
    |C(x) - \wt{C}_R(x)| \leq \epsilon + C(x+2\delta) - C(x),
\end{align}
which gives
\begin{align}
    \wt{C}(x + \delta) \leq C(x + 2\delta) +\epsilon.
\end{align}

The lemma is then proved.
\end{proof}
\subsubsection{Technical lemma}
\begin{lemma}[Mollifier, Lemma 5 in \cite{lt21}]\label{lem:mollifier}
Define $M_{d,\delta}(x)$ to be
\begin{align}
    M_{d,\delta}:=\frac{1}{{\cal N}_{d,\delta}}T_d\left( 1 + 2\frac{\cos(x) - \cos (\delta)}{1+\cos (\delta)}\right)
\end{align}
where $T_d(x)$ is the $d$-th Chebyshev polynomial of the first kind, and
\begin{align}
    {\cal N}_{d,\delta}:=\int_{-\pi}^{\pi} T_d\left( 1 + 2\frac{\cos(x) - \cos (\delta)}{1+\cos (\delta)}\right) \d x.
\end{align}
Then
\begin{enumerate}
    \item $|M_{d,\delta}(x)|\leq \frac{1}{{\cal N}_{d,\delta}}$ for $x\in [-\pi, -\delta]\cup [\delta, \pi]$, and $M_{d,\delta}(x)\geq \frac{1}{{\cal N}_{d,\delta}}$ for $x\in [-\delta, \delta]$.
    \item $\int_{-\pi}^{\pi} M_{d,\delta}(x)\d x = 1$, $1\leq \int_{-\pi}^{\pi} |M_{d,\delta}(x)| \d x \leq 1+\frac{4\pi}{{\cal N}_{d,\delta}}$.
    \item When $\tan (\delta/2)\leq 1-1/\sqrt{2}$, we have
    \begin{align}
        {\cal N}_{d,\delta} \geq C_1 e^{d\delta/\sqrt{2}}\sqrt{\frac{\delta}{d}}\cdot \mathrm{erf}(C_2\sqrt{d}\delta),
    \end{align}
    for some universal constant $C_1,C_2$.
\end{enumerate}
\end{lemma}
The proof can be found in Appendix A in \cite{lt21}, and we omit it here.
\section{Technical details of the Hadamard test of block-encoded observable}
\label{sec:hadamard_test_block_encoding}
In this section, we give detailed analysis of the Hadamard test for block-encodings which plays a crucial role in the proof of Theorem \ref{thm:gspe_complexity_general_case}.

We first note that the quantum state before the final measurements is as follows: 
\begin{align}
    \ket{\phi_1}=\begin{cases}
    \frac{1}{\sqrt{2}}\left(\ket{+}\ket{0^m}\ket{\phi_0} + \ket{-}(I\otimes e^{-iHt_2})U(I\otimes e^{-iH t_1})\ket{0^m}\ket{\phi_0}\right)& \text{if}~W=I,\\
    \frac{1}{\sqrt{2}}\left(\ket{+}\ket{0^m}\ket{\phi_0} + i\ket{-}(I\otimes e^{-iHt_2})U(I\otimes e^{-iH t_1})\ket{0^m}\ket{\phi_0}\right)& \text{if}~W=S.
    \end{cases}
\end{align}

\paragraph{Case 1: $W=I$}

We measure the first two registers. If the outcome is $(0, 0^m)$, the (un-normalized) remaining state is:
\begin{align}
    &(\bra{0}\bra{0^m}\otimes I)\frac{1}{\sqrt{2}}\left(\ket{+}\ket{0^m}\ket{\phi_0} + \ket{-}(I\otimes e^{-iHt_2})U(I\otimes e^{-iH t_1})\ket{0^m}\ket{\phi_0}\right)\notag\\
    = &~ \frac{1}{2}\ket{\phi_0}+\frac{1}{2\alpha}e^{-iHt_2}Oe^{-iHt_1}\ket{\phi_0}
\end{align}
Hence, this event happens with the following probability:
\begin{align}
    &\Pr[\text{the outcome is }(0,0^m)|W=I]\\
    =&~ \bra{\phi_0}\left(\frac{1}{2}I+\frac{1}{2\alpha}e^{iHt_1}O^\dagger e^{iHt_2}\right)\left(\frac{1}{2}I+\frac{1}{2\alpha}e^{-iHt_2}Oe^{-iHt_1}\right)\ket{\phi_0}\notag\\
    =&~ \frac{1}{4}\left(1+\frac{1}{\alpha}\bra{\phi_0}e^{iHt_1}O^\dagger e^{iHt_2}\ket{\phi_0}+\frac{1}{\alpha}\bra{\phi_0}e^{-iHt_2}O e^{-iHt_1}\ket{\phi_0} + \frac{1}{\alpha^2}\bra{\phi_0}e^{iHt_1}O^\dagger Oe^{-iHt_1}\ket{\phi_0}\right).
\end{align}
Similarly, if the outcome is $(1, 0^m)$, the remaining (un-normalized) state is
\begin{align}
    \frac{1}{2}\ket{\phi_0}-\frac{1}{2\alpha}e^{-iHt_2}Oe^{-iHt_1}\ket{\phi_0},
\end{align}
and the probability is
\begin{align}
    &\Pr[\text{the outcome is }(1,0^m)|W=I]\notag\\
    = &~ \frac{1}{4}\left(1-\frac{1}{\alpha}\bra{\phi_0}e^{iHt_1}O^\dagger e^{iHt_2}\ket{\phi_0}-\frac{1}{\alpha}\bra{\phi_0}e^{-iHt_2}O e^{-iHt_1}\ket{\phi_0} + \frac{1}{\alpha^2}\bra{\phi_0}e^{iHt_1}O^\dagger Oe^{-iHt_1}\ket{\phi_0}\right).
\end{align}
Hence, the expectation of $X$ is
\begin{align}
    \E[X]=&~ \alpha \cdot (\Pr[\text{the outcome is }(0,0^m)|W=I] - \Pr[\text{the outcome is }(1,0^m)|W=I])\\
    = &~ \frac{1}{2}(\bra{\phi_0}e^{-iHt_2}O e^{-iHt_1}\ket{\phi_0} + \bra{\phi_0}e^{iHt_1}O^\dagger e^{iHt_2}\ket{\phi_0})\notag\\
    = &~ \frac{1}{2}(\bra{\phi_0}e^{-iHt_2}O e^{-iHt_1}\ket{\phi_0} + \ov{\bra{\phi_0}e^{-iHt_2}O e^{-iHt_1}\ket{\phi_0}})\notag\\
    = &~ \Re \bra{\phi_0}e^{-iHt_2}O e^{-iHt_1}\ket{\phi_0}.
\end{align}

\paragraph{Case 2: $W=S$} Similar to the case 1, we have
\begin{align}
    & \Pr[\text{the outcome is }(0,0^m)|W=S]\\
    =&~ \bra{\phi_0}\left(\frac{1}{2}I-\frac{i}{2\alpha}e^{iHt_1}O^\dagger e^{iHt_2}\right)\left(\frac{1}{2}I+\frac{i}{2\alpha}e^{-iHt_2}Oe^{-iHt_1}\right)\ket{\phi_0}\notag\\
    =&~ \frac{1}{4}\left(1-\frac{i}{\alpha}\bra{\phi_0}e^{iHt_1}O^\dagger e^{iHt_2}\ket{\phi_0}+\frac{i}{\alpha}\bra{\phi_0}e^{-iHt_2}O e^{-iHt_1}\ket{\phi_0} + \frac{1}{\alpha^2}\bra{\phi_0}e^{iHt_1}O^\dagger Oe^{-iHt_1}\ket{\phi_0}\right).
\end{align}
And
\begin{align}
    &\Pr[\text{the outcome is }(1,0^m)|W=S]\notag\\
    = &~ \frac{1}{4}\left(1+\frac{i}{\alpha}\bra{\phi_0}e^{iHt_1}O^\dagger e^{iHt_2}\ket{\phi_0}-\frac{i}{\alpha}\bra{\phi_0}e^{-iHt_2}O e^{-iHt_1}\ket{\phi_0} + \frac{1}{\alpha^2}\bra{\phi_0}e^{iHt_1}O^\dagger Oe^{-iHt_1}\ket{\phi_0}\right).
\end{align}
Hence,
\begin{align}
    \E[Y]=&~ \alpha \cdot (\Pr[\text{the outcome is }(1,0^m)|W=S] - \Pr[\text{the outcome is }(0,0^m)|W=S])\\
    = &~ \frac{i}{2}(-\bra{\phi_0}e^{-iHt_2}O e^{-iHt_1}\ket{\phi_0} + \ov{\bra{\phi_0}e^{-iHt_2}O e^{-iHt_1}\ket{\phi_0}})\notag\\
    = &~ \Im \bra{\phi_0}e^{-iHt_2}O e^{-iHt_1}\ket{\phi_0}.
\end{align}

Therefore,
\begin{align}
    \E[X+iY] = \bra{\phi_0}e^{-iHt_2}O e^{-iHt_1}\ket{\phi_0}.
\end{align}

\subsection{Generalized Hadamard test}
In this subsection, we study the generalized Hadamard test for block-encodings and we will show that the estimator's variance can be reduced by replacing the first Hadamard gate with an $\alpha$-dependent single-qubit gate.

Suppose $W=I$ and we replace the first Hadamard gate with the following single-qubit gate:
\begin{align}
    G(a, b, \theta):=\begin{bmatrix}
    a & b\\
    -e^{i\theta}\ov{b} & e^{i\theta}\ov{a}
    \end{bmatrix},
\end{align}
where $\theta\in \R$, $a,b\in \C$ with $|a|^2 + |b|^2=1$.

Then, we have
\begin{align*}
    &\ket{0}\ket{0^m}\ket{\phi_0}\\
    \xrightarrow{G(a, b, \theta)}&~ a\ket{0}\ket{0^m}\ket{\phi_0}  -e^{i\theta}\ov{b} \ket{1}\ket{0^m}\ket{\phi_0}\\
    \xrightarrow{C\mhyphen e^{-iHt_1}}&~ a\ket{0}\ket{0^m}\ket{\phi_0} -e^{i\theta}\ov{b} \ket{1}(I\otimes e^{-iHt_1})\ket{0^m}\ket{\phi_0}\\
    \xrightarrow{C\mhyphen U}&~ a\ket{0}\ket{0^m}\ket{\phi_0} -e^{i\theta}\ov{b} \ket{1}U(I\otimes e^{-iHt_1})\ket{0^m}\ket{\phi_0}\\
    \xrightarrow{C\mhyphen e^{-iHt_2}}&~ a\ket{0}\ket{0^m}\ket{\phi_0} -e^{i\theta}\ov{b} \ket{1}(I\otimes e^{-iHt_2})U(I\otimes e^{-iHt_1})\ket{0^m}\ket{\phi_0}\\
    \xrightarrow{G(p,q,\rho)}&~ a(p\ket{0} - e^{i\rho}\ov{q}\ket{1})\ket{0^m}\ket{\phi_0} -e^{i\theta}\ov{b} (q\ket{0} + e^{i\rho}\ov{p}\ket{1})(I\otimes e^{-iHt_2})U(I\otimes e^{-iHt_1})\ket{0^m}\ket{\phi_0}\\
    =: &~ \ket{\phi_1}.
\end{align*}
Hence, the un-normalized remaining state after the measurement with outcome $(0, 0^m)$ is:
\begin{align}
    (\bra{0}\bra{0^m}\otimes I)\ket{\phi_1}
    = ap \ket{\phi_0} - \frac{e^{i\theta}\ov{b}q}{\alpha} e^{-iHt_2}Oe^{-iHt_1}\ket{\phi_0}.
\end{align}
It implies that
\begin{align}
    &\Pr[\text{the outcome is }(0, 0^m)|W=I]\\
    = &~ \bra{\phi_0} \left(\ov{a}\ov{p} I - \frac{e^{-i\theta}b\ov{q}}{\alpha} e^{iHt_1}O^\dagger e^{iHt_2}\right) \left(ap I - \frac{e^{i\theta}\ov{b}q}{\alpha} e^{-iHt_2} O e^{-iHt_1}\right)\ket{\phi_0}\notag\\
    = &~ |a|^2|p|^2 +  \frac{|b|^2|q|^2}{\alpha^2}\bra{\phi_0}e^{iHt_1} O^\dagger O e^{-iHt_1}\ket{\phi_0}\notag\\
    -& \frac{e^{i\theta}\ov{a}\ov{b}\ov{p}q}{\alpha} \bra{\phi_0}e^{-iHt_2} O e^{-iHt_1}\ket{\phi_0} - \frac{e^{-i\theta}abp\ov{q}}{\alpha} \ov{\bra{\phi_0}e^{-iHt_2} O e^{-iHt_1}\ket{\phi_0}}.
\end{align}

On the other hand, the un-normalized state for the outcome $(1, 0^m)$ is
\begin{align}
    (\bra{1}\bra{0^m}\otimes I)\ket{\phi_1}
    = -e^{i\rho}a\ov{q} \ket{\phi_0} - \frac{e^{i(\theta+\rho)}\ov{b}\ov{p}}{\alpha} e^{-iHt_2}Oe^{-iHt_1}\ket{\phi_0},
\end{align}
and the probability is
\begin{align}
    &\Pr[\text{the outcome is }(1, 0^m)|W=I]\\
    = &~ \bra{\phi_0} \left(-e^{-i\rho}\ov{a}q I - \frac{e^{-i(\theta+\rho)}bp}{\alpha} e^{iHt_1}O^\dagger e^{iHt_2}\right) \left(-e^{i\rho}a\ov{q} I - \frac{e^{i(\theta+\rho)}\ov{b}\ov{p}}{\alpha} e^{-iHt_2} O e^{-iHt_1}\right)\ket{\phi_0}\notag\\
    = &~ |a|^2|q|^2 + \frac{|b|^2|p|^2}{\alpha^2}\bra{\phi_0}e^{iHt_1} O^\dagger O e^{-iHt_1}\ket{\phi_0}\notag\\
    +&\frac{e^{i\theta}\ov{a}\ov{b}\ov{p}q}{\alpha}\bra{\phi_0}e^{-iHt_2} O e^{-iHt_1}\ket{\phi_0} + \frac{e^{-i\theta}abp\ov{q}}{\alpha}\ov{\bra{\phi_0}e^{-iHt_2} O e^{-iHt_1}\ket{\phi_0}}
\end{align}

If we choose $|p|=|q|=\frac{1}{\sqrt{2}}$, then we have
\begin{align}
    &\Pr[\text{the outcome is }(1, 0^m)|W=I] - \Pr[\text{the outcome is }(0, 0^m)|W=I]\notag\\
    = &~ \Re \frac{4e^{i\theta} \ov{a}\ov{b}\ov{p}q}{\alpha}\bra{\phi_0}e^{-iHt_2} O e^{-iHt_1}\ket{\phi_0}.
\end{align}
Notice that to make the Hadamard test work, we need the coefficient $\frac{4e^{i\theta} \ov{a}\ov{b}\ov{p}q}{\alpha}$ to be a real or an imaginary number. 

Now, we show how to choose the parameters to minimize the variance. Without loss of generality, we may assume $a,b\in (0,1)$ such that $a^2+b^2=1$ and use $p,q$ to cancel the phase factor, i.e., $e^{i\theta} \ov{a}\ov{b}\ov{p}q=\frac{1}{2}ab$. It gives that:
\begin{align}
    &\Pr[\text{the outcome is }(1, 0^m)|W=I] - \Pr[\text{the outcome is }(0, 0^m)|W=I]\notag\\
    =&~ \frac{2ab}{\alpha}\Re \bra{\phi_0}e^{-iHt_2} O e^{-iHt_1}\ket{\phi_0},
\end{align}
and
\begin{align}
    &\Pr[\text{the outcome is }(1, 0^m)|W=I] + \Pr[\text{the outcome is }(0, 0^m)|W=I]\notag\\
    = &~ a^2 + \frac{b^2}{\alpha^2}\bra{\phi_0}e^{iHt_1}O^\dagger Oe^{-iHt_1}\ket{\phi_0} 
\end{align}

Now, define the random variable as follows:
\begin{align}
    X:=\begin{cases}
    \frac{\alpha}{2ab} & \text{if the outcome is }(1, 0^m),\\
    -\frac{\alpha}{2ab} & \text{if the outcome is }(0, 0^m),\\
    0 & \text{otherwise}.
    \end{cases}
\end{align}
Then, we have
\begin{align}
    \E[X]=\Re \bra{\phi_0}e^{-iHt_2} O e^{-iHt_1}\ket{\phi_0}.
\end{align}
And we have
\begin{align}
    \Var[X]=&~ \E[X^2]-\E[X]^2\notag\\
    =&~ \frac{\alpha^2}{4a^2b^2}\left(a^2 + \frac{b^2}{\alpha^2}\bra{\phi_0}e^{iHt_1}O^\dagger Oe^{-iHt_1}\ket{\phi_0} \right) - \left(\Re \bra{\phi_0}e^{-iHt_2} O e^{-iHt_1}\ket{\phi_0}\right)^2.
\end{align}
The second term is fixed for any parameters. And for the first term, we have
\begin{align}
    \frac{\alpha^2}{4a^2b^2}\left(a^2 + \frac{b^2}{\alpha^2}\bra{\phi_0}e^{iHt_1}O^\dagger Oe^{-iHt_1}\ket{\phi_0} \right) = &~ \frac{\alpha^2}{4b^2} + \frac{1}{4a^2}\bra{\phi_0}e^{iHt_1}O^\dagger Oe^{-iHt_1}\ket{\phi_0}\\
    = &~ \frac{\alpha^2}{4(1-a^2)} + \frac{\|Oe^{-iHt_1}\ket{\phi_0}\| ^2}{4a^2}\notag\\
    \geq &~ \frac{1}{4}(\alpha + \|Oe^{-iHt_1}\ket{\phi_0}\|)^2,
\end{align}
where the minimizer is at $a:=\sqrt{\frac{\|Oe^{-iHt_1}\ket{\phi_0}\|}{\alpha + \|Oe^{-iHt_1}\ket{\phi_0}\|}}$. However, since we do not know the value of $\|Oe^{-iHt_1}\ket{\phi_0}\|$, there are two approaches to resolve this issue: (1) use another quantum circuit to estimate $\|Oe^{-iHt_1}\ket{\phi_0}\|$ and then set the parameters; (2) just take $a:=\sqrt{\frac{1}{\alpha+1}}$. Notice that when the first gate is the Hadamard gate, i.e., $a=\frac{1}{\sqrt{2}}$, we have
\begin{align}
    \Var\left[X~\Big|~a=\frac{1}{\sqrt{2}}\right] = \frac{1}{2}(\alpha^2 + \|Oe^{-iHt_1}\ket{\phi_0}\|^2).
\end{align}
When $a=\sqrt{\frac{1}{\alpha+1}}$, we have
\begin{align}
    \Var\left[X~\Big|~a=\frac{1}{\sqrt{\alpha+1}}\right]=&~ \frac{1}{4}\alpha(\alpha+1) + \frac{1}{4}\|Oe^{-iHt_1}\ket{\phi_0}\|^2(\alpha+1)\\
    = &~ \frac{1}{2}(\alpha^2 + \|Oe^{-iHt_1}\ket{\phi_0}\|^2) -\frac{1}{4}(\alpha-1)(\alpha-  \|Oe^{-iHt_1}\ket{\phi_0}\|^2)\notag\\
    \leq &~ \Var\left[X~\Big|~a=\frac{1}{\sqrt{2}}\right],
\end{align}
where the last step follows from $\alpha\geq 1$ and $\|Oe^{-iHt_1}\ket{\phi_0}\|^2\leq 1$. Therefore, we can reduce the estimator's variance by choosing $a=\sqrt{\frac{1}{\alpha}}$. Moreover, if $\alpha$ is large, the new variance is about half of the variance using the Hadamard gate. 

Similar strategy can also be used to reduce the variance of the random variable $Y$.

\end{document}